\documentclass[12pt,a4paper]{article}
\usepackage{amssymb,amsmath,amsthm}
\usepackage{wasysym}
\usepackage[colorlinks=true,citecolor=blue]{hyperref}
\usepackage[normalem]{ulem}
\usepackage{color,graphicx}
\usepackage{stmaryrd, pxfonts}
\newcommand{\x}{$\bullet$}
\newcommand{\ci}{$\circ$}
\newcommand{\LL}{\triangleleft}
\newcommand{\dd}{drum }
\newcommand{\CH}{\textrm{CH}}

\newcommand{\ord}{\ensuremath\triangleleft}

\newtheorem{theorem}{Theorem}
\newtheorem{lemma}[theorem]{Lemma}

\newtheorem{observation}[theorem]{Observation}

\newtheorem{corollary}[theorem]{Corollary}

\theoremstyle{definition}
\newtheorem{definition}[theorem]{Definition}
\theoremstyle{remark}
\newtheorem*{remark}{Remark}

\renewcommand{\qedsymbol}{\smiley}

\usepackage{comment}

{\count0=\time \divide\count0 by 60 \count2=\count0
\multiply \count2 by -60 \advance \count2 by \time
\ifnum\count2<10
\xdef\urzeit{\the\count0:0\the\count2}\else
\xdef\urzeit{\the\count0:\the\count2}\fi}

\usepackage{fullpage}

\title{Quasi-Parallel Segments and Characterization of Unique
  Bichromatic Matchings\footnote
{Supported by by the ESF EUROCORES programme EuroGIGA, CRP ComPoSe, Deutsche Forschungsgemeinschaft (DFG), under grant FE 340/9-1.}}
\author{Andrei Asinowski\footnote{Institut f\"ur Informatik, Freie
    Universit\"at Berlin, Germany. \texttt{asinowski@gmail.com}} \\
  Tillmann Miltzow\footnote{Institut  f\"ur Informatik, Freie Universit\"at Berlin, Germany. \texttt{t.miltzow@gmail.com}}\\   G\"{u}nter Rote\footnote{Institut f\"ur Informatik, Freie Universit\"at Berlin, Germany. \texttt{rote@inf.fu-berlin.de}}}

\begin{document}
\maketitle

\begin{abstract}
Given $n$ red and $n$ blue points in general position in the plane,
it is well-known that there is a perfect matching formed by
non-crossing line segments.
We characterize the bichromatic point sets which admit
exactly one
non-crossing matching.
We give several geometric descriptions of such sets,
and find an $O(n \log n)$ algorithm that checks
whether a given bichromatic set has this property.
\end{abstract}

%%%%%%%%%%%%%%%%%%%%%%%%%%%%%%%%%%%%%%%%%%%%%%%%%%%%%%%%%%%%%%%%%%%%%%%%%%%
\section{Introduction}
\label{sec:intro}
\paragraph{Basic notation and preliminary results}

A \emph{bichromatic $(n+n)$ point set} $F$ is a set of $n$ blue points and $n$ red points in the plane.
 We assume that the points of $F$ are in general position,
this is, no three points lie on the same line.
A \emph{perfect bichromatic non-crossing straight-line matching} of $F$
is a set of $n$ non-crossing segments between points of $F$
%so that each segment has one blue and one red endpoint,
%and 
so that each blue point is connected to exactly one red point, and vice versa.
Following the convention in the literature, we call such matchings \emph{BR-matchings}.
For notational reasons, we shall denote the colors blue and red
by \textit{\ci}
(white) and 
\textit{\x}
(black). 

It is well known and easy to see that any $F$ has at least one BR-matching.
One way to see this is to use recursively the Ham-Sandwich Theorem
(%cf.~\cite{m-ldc-02,or-some-other-textbook},
see Theorem~\ref{theorem-algo} below);
another way is to show that the bichromatic matching that minimizes
the total length of segments is necessarily non-crossing.
The main motivation of our work is
%Therefore, it is interesting 
to characterize bichromatic sets with a \emph{unique} BR-matching.
We will establish connections between this question and various other
geometric notions that have shown up in different contexts.

In what follows, $M$ denotes a given BR-matching of $F$.
The segments in $M$ are considered directed from the \ci-end to the \x-end.
For $A \in M$, the line that contains $A$ is denoted by $g(A)$,
and it is considered directed consistently with $A$.
For two directed segments $A$ and $B$ for which the lines $g(A)$ and
$g(B)$ do not cross, we
we say that the segments (resp., the lines) are \emph{parallel}
if they have the same orientation;
otherwise we call them
\emph{antiparallel}.
If we delete inner points of $A$ from $g(A)$, we obtain two closed \emph{outer rays}:
the \emph{\ci-ray} and the \emph{\x-ray}, according to the endpoint of $A$ that belongs to the ray.

The convex hull of $F$ will be denoted by $\CH(F)$,
and % the points of $F$ on
its boundary by $\partial \CH(F)$.
Consider the circular sequence of colors of the points of $\partial \CH(F)$;
a \emph{color interval} is a maximal subsequence of this circular sequence that consists of points of the same color.
In the point set in Fig.~\ref{fig:def}(a), $\partial \CH(F)$ has four color intervals:
two \ci-intervals (of sizes $1$ and $2$) and two \x-intervals (of sizes $2$ and $3$).

\begin{figure}[h]
		\begin{center}
%\resizebox{140mm}{!}{\input{figures/def3.pdf_t}}
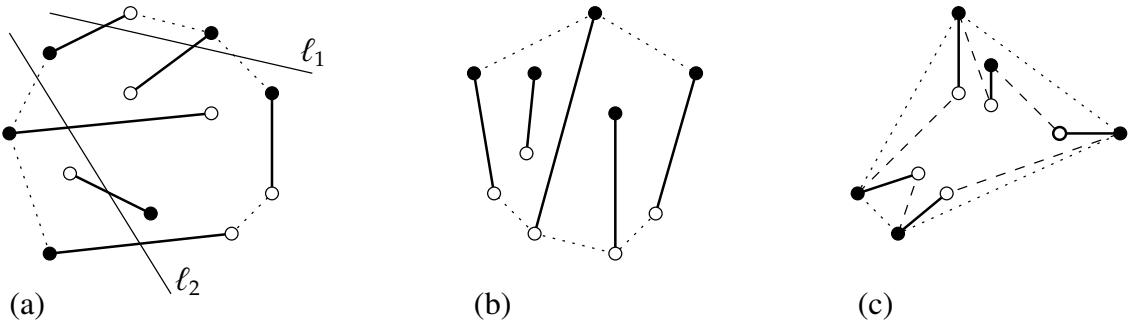
		\caption{(a) A matching with chromatic cuts;
		(b) A linear matching; (c) A circular matching
		(another matching for the same point set is indicated by dashed lines).}
		\label{fig:def}

			\end{center}
	\end{figure}

%\textbf{GR: should we use arrows for the edges, to indicate the directions?}
%\texttt{AAA: I don't think that using arrows is a good idea.
%The figures will be more complicated.
%Besides, we don't really use directions, but rather only want to say that some two lines intersect in rays of the same color etc.
%Colors themselves suffice for this matter.}
%\texttt{T: I agree with AAA for the same reasons and comment this discussion}

In order to state our main result, we need the notion of chromatic cut.
A \emph{chromatic cut} of $M$ is a line $\ell$ that crosses
two segments of $M$ such that their \x-ends are on different sides of $\ell$
($\ell$ can as well cross other segments of $M$).
For example,
 the lines $\ell_1$ and $\ell_2$ in Fig.~\ref{fig:def}(a) are chromatic cuts.
The matchings in Fig.~\ref{fig:def}(b) and (c) have no chromatic cuts.
Aloupis et~al.~\cite[Lemma~9]{abls-bcm-13} proved that
 a BR-matching $M$ that has a chromatic cut cannot be unique.
They actually proved a stronger statement: there is even a
\emph{compatible} BR-matching $M'\neq M$,
which means that the union of $M$ and $M'$ is non-crossing. % set of segments).
Thus, having no chromatic cut
is a \emph{necessary condition} for a unique BR-matching.
(We will give a simpler proof of this fact in Section~\ref{sec:cc},
without
establishing the existence of a compatible matching.)
However, it is \emph{not sufficient},
as shown by the example in Fig.~\ref{fig:def}(c).

We will give a thorough treatment of
BR-matchings without chromatic cuts.
We shall prove in Lemma~\ref{lem:otherPS_v2} that
BR-matchings without chromatic cuts
can be classified into
 the following two types.
A \emph{matching of linear type} (or, for shortness, \emph{linear matching}) is a BR-matching without a chromatic cut
such that $\partial \CH(F)$ consists of exactly two color intervals
(both necessarily of size at least $2$).
A \emph{matching of circular type} (or \emph{circular matching}) is a BR-matching without a chromatic cut
such that all points of $\partial \CH(F)$ have the same color. The reason for these terms will be clarified below.
 Fig.~\ref{fig:def}(b--c) shows a linear and circular matching.

We shall prove that the unique BR-matchings are precisely
the linear matchings. This will be a part of our main result,
 Theorem~\ref{thm:main} below.

\paragraph{The main result.}
We formulate the main result in the following three Theorems.
\begin{theorem}\label{the:two_types}
Let $M$ be a BR-matching without a chromatic cut. 
Then $M$ is either of linear or circular type.
\end{theorem}

 We will characterize both linear and circular matchings. 

The following series of equivalent characterizations of unique BR-matchings.
 Condition~\ref{order} refers to a {relation $\triangleleft$}
which is defined later in Definition~\ref{def:orderM}.
The notion of a \emph{quasi-parallel matching} in
 Condition~\ref{qp} will be defined in Definition~\ref{def:quasiParallel}.
We state the conditions here to have all
 equivalent characterizations in one place.

\begin{theorem}
[Characterization of unique BR-matchings]
\label{thm:main}
	Let $M$ be a BR-matching of $F$.
2	Then the following conditions are equivalent:
\begin{enumerate}
	\item $M$ is the only BR-matching of~$F$.
	\item \label{type_l} $M$ is a linear matching.
	%\item \label{intervals}
%$\partial \CH(F)$ consists of two color intervals of length at least $2$;
		%and for any $A, B \in M$ ($A \neq B$),
		%either $g(A)$ and $g(B)$ are parallel (including orientation),
		%or their intersection point belongs to the outer rays of the same color.
\item \label{order} The relation \ord\ is a linear order on $M$.
\item \label{forbidden} No subset of segments forms one of the three
       forbidden patterns shown below in Fig.~\ref{fig:allConfig}.
		\item \label{qp} $M$ is quasi-parallel.
%		\item Any submatching $M'$ satisfies one of the above conditions.
%	\item Any submatching of $M$ satisfies the above conditions. % : M' denotes everywhere "another matching"
	\end{enumerate}
\end{theorem}

Moreover, if $M$ satisfies
any of
 the above conditions, then
 any submatching of $M$ satisfies the above conditions.
This follows from the fact that conditions
\ref{order}--\ref{qp} directly imply that they hold for all subsets.
(This is most trivial for condition~\ref{forbidden}.)

\begin{figure}[h]
	\begin{center}
		\includegraphics{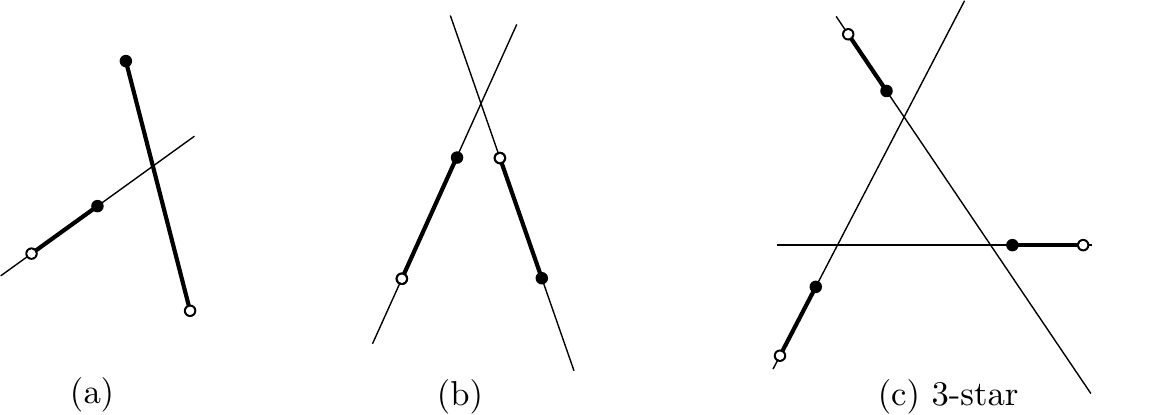}
		\caption{Forbidden patterns for quasi-parallel matchings.}
		\label{fig:allConfig}
	\end{center}
\end{figure}

\begin{theorem}[Properties of circular matchings]\label{thm:circularMatchings}
	Let $M$ be a BR-matching of $F$. Then the following conditions are equivalent:
	\begin{enumerate}
		\item $M$ is a circular matching.
		\item The sidedness relation \ord\ is total but not a
                  linear order.
		\item No subset of segments forms one of the forbidden patterns shown
			in Fig.~\ref{fig:allConfig}~(a)-(b) and at least three segments form
			a $3$-star as in Fig.~\ref{fig:allConfig}~(c).
	\end{enumerate}
 Furthermore, if these conditions hold, then:
	\begin{enumerate}
		\item[p1.] The sidedness relation \ord\ induces naturally a circular order, explained further in Section~\ref{sec:type_c}.
		\item[p2.] There are (at least) two additional disjoint BR-matchings $M'$ and $M''$ on $F$.
	\end{enumerate}
\end{theorem}
Property~p1 justifies the term circular matching.

\begin{center}
\begin{tabular}{|c||c|c|}
\hline
 & Linear Type & Circular Type  \\ \hline
Uniqueness  & $M$ is unique & $M$ is not unique \\ \hline
Patterns from Fig.~\ref{fig:allConfig}  & (a), (b) and (c) are avoided & 
\begin{tabular}{c}
(a) and (b) are avoided;  \\ (c) is present  
\end{tabular} \\ \hline
Relation $\triangleleft$  & Linear order & 
\begin{tabular}{c}
Total, not linear; \\ induces a circular order 
\end{tabular}
\\ \hline
\end{tabular}
\end{center}

\paragraph{Related work.}
%The general framework of this work are geometric augmentations.
%A recent survey is written by Ferran Hurtado and Csaba D. T\'oth~\cite{FerranToth2003survey}.
%Given some geometric graph with certain properties one might ask how many segments do one has to add in order to attain a certain property.
%Or under which conditions can we find at all a set of segments to attain some property.
Monochromatic and bichromatic straight-line matchings have been intensively studied in the recent years.

One direction is \emph{geometric augmentation}.
Given a matching,
one wants to determine whether it is possible to add segments in order to
get a bigger matching with a certain property,
under what conditions can this be done,
how many segments one has to add, etc.
See %Ferran
 Hurtado and Cs.~%aba
 D.~T\'oth~\cite{FerranToth2003survey}
for a recent survey. 

The \emph{bichromatic compatible matching graph} of $F$ has at its node set the BR-matchings of $F$.
Two BR-matchings are joined by an edge if they are compatible. %, that is, if their union is crossing-free.
% Irrelevant
% Note that a BR-matching is compatible to itself.
%It was
% By the time of publication, it won't be *very* recently.
%very
%
%Greg
 Aloupis, % Luis
 Barba, % Stefan
 Langerman and % Diane L.
 Souvaine~\cite{abls-bcm-13}
proved recently that
the
bichromatic compatible matching graph is
always connected.

%When the color condition on the matchings is omitted we receive the \emph{(monochromatic) compatible matching graph}.
For non-colored point sets, one can speak about the \emph{(monochromatic) compatible matching graph}.
The diameter of this graph is $O(\log n)$ \cite{AichCompMatch2009}, whereas for the bichromatic
compatibility graph, no non-trivial upper bound is known.
%
%Alfredo
 Garc\'{\i}a, %Marc
 Noy, and % Javier
 Tejel~\cite{Garcia:2000:LBN:353029.353031}
 showed that the number of perfect monochromatic matchings is
 minimized among all $n$-point sets when the points are in convex position.
%
%Mashhood
 Ishaque, % Diane L.
 Souvaine, and Cs.~%aba D.
 T\'oth~\cite{Ishaque:2011:DCG:1998196.1998216} showed that for any
 monochromatic perfect matching, there is a \emph{disjoint} monochromatic compatible matching.

Other related work involves counting the maximal number of BR-matchings that a point set admits.
%Micha
 Sharir and % Emo
 Welzl~\cite{sw-ncfm-06j} established a bound of $O(7.61^n)$.

\paragraph{Outline.}
In Section~\ref{sec:cc} we prove in a simple way that if a BR-matching
$M$ has a chromatic cut, then it is not unique. % BR-matching of $F$.
In Section~\ref{sec:type_l} we give different characterizations of
 linear matchings, and we prove that a linear matching is
 unique. % for $F$.
In Section~\ref{sec:type_c} we analyze circular matchings, and prove
that they are never unique. % for its point set.
In Section~\ref{sec:summary} we complete the proof of Theorem~\ref{thm:main}.
In Section~\ref{sec:misc} we discuss parallelizability of linear matchings
and enumerate sidedness relations realizable by $n$ element circular matchings.
In Section~\ref{sec:algo} we describe an algorithm that recognizes point sets $F$ which admit exactly one matching,
an algorithm that recognizes circular matchings, 
and some more algorithms that compute some other notions we used in the previous Sections.
We conclude by mentioning some open problems and possible directions for future research in Section~\ref{sec:open}.

%%%%%%%%%%%%%%%%%%%%%%%%%%%%%%%%%%%%%%%%%%%%%%%%%%%%%%%%%%%%%%%%%%%%%%%%%%%
\section{Chromatic cuts} \label{sec:cc}
%%%%%%%%%%%%%%%%%%%%%%%%%%%%%%%%%%%%%%%%%%%%%%%%%%%%%%%%%%%%%%%%%%%%%%%%%%%%

We start with a simple geometric description of BR-matchings that admit a chromatic cut.

\begin{figure}[h]
			\begin{center}
			\includegraphics[width=100mm]{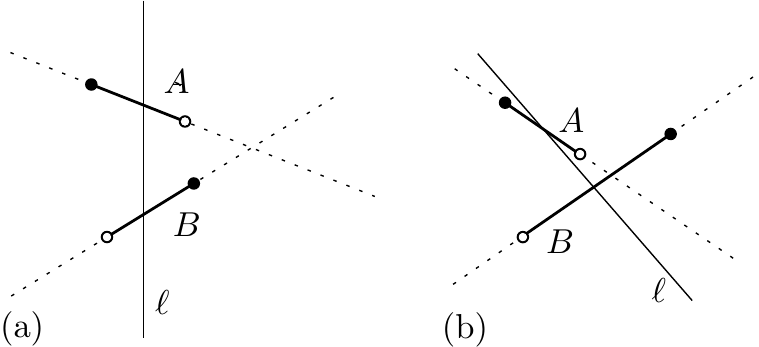}
				\caption{(a) Outer rays of different colors intersect;
				(b) An outer ray crosses the second segment.}
				\label{fig:charCCuts}
			\end{center}
	\end{figure}

	\begin{lemma}\label{lem:chCCut}
Let $M$ be a BR-matching of $F$.
$M$ admits a chromatic cut if and only if
there are two segments $A, B\in M$ such that
%either
 $A$ and $B$ are antiparallel,
or the intersection point of $g(A)$ and $g(B)$ belongs to outer rays of different colors,
or an outer ray of one of the segments crosses the second segment.
\end{lemma}
\begin{proof} ``$\Leftarrow$''
If the \x-ray of one segment and the \ci-ray of the second segment cross each other,
then any line through inner points of $A$ and $B$ is a chromatic cut, see Fig.~\ref{fig:charCCuts}~(a).
The same is true if $A$ and $B$ are antiparallel.
If an outer ray of one of the segments crosses the second segment,
%(without loss of generality the \x-ray of $A$ crosses $B$),
%without loss of generality, we assume that the positive ray determined by $a$ crosses $b$,
%and that the \x-end of $b$ is in the right half-plane bounded by $\ell_a$;
%see Figure~\ref{fig:charCCuts}~c).
%Then if we rotate $\ell_a$ clockwise about an inner point of $a$ by a small angle,
then, if we rotate $g(A)$ around an inner point of $A$ by a small angle in one of two possible directions
(depending on the orientation of $A$ and $B$),
then a chromatic cut is obtained, see Fig.~\ref{fig:charCCuts}~(b).

%Consider the line $l_c$ through the midpoint of $a$ tilted slightly in clockwise direction and $l_{cc}$ tilted slightly in counterclockwise direction both crossing $b$. It is clear that the \x-point of $b$ is on %the same side of both lines whereas the \x-point of $a$ is on different sides of the lines.

``$\Rightarrow$'' Let $\ell$ be a chromatic cut of $M$,
 and let $A$ and $B$ be two segments that have their \x-ends on the opposite sides of $\ell$.
Consider the lines $g(A)$ and $g(B)$.
 If $g(A)$ and $g(B)$ don't cross, they clearly must be antiparallel.
% Assume now that $g(A)$ and $g(B)$
If they cross, then
it is not possible that the two
%\x-rays or the two \ci-rays meet,
outer rays of the same color meet,
because they are on opposite sides of $\ell$.
 \end{proof}

A line $\ell$ is a \emph{balanced line}
if in each open halfplane determined by $\ell$,
the number of \x-points is equal to the number of \ci-points.
We say that an open halfplane is \emph{dominated} by \x-points
if it contains more \x-points than \ci-points.
%Chromatic cuts and balanced lines are more closely related then just their similar definition.
The next lemma reveals a relation between chromatic cuts and balanced lines.

\begin{lemma}\label{lem:cutBal}
Let $M$ be a BR-matching.
$M$ has a chromatic cut if and only if
there exists a balanced line that crosses a segment of $M$.
\end{lemma}
\begin{proof}
	``$\Leftarrow$'' %A balanced line $\ell$ cannot cross exactly one segment of $M$. THIS IS WRONG!
Let $\ell$ be a balanced line that crosses a segment $A$.
We can assume that $\ell$ does not contain points from $F$:
it cannot contain exactly one point of $F$;
and if it contains two points of $F$ of different colors, we can
translate
 it slightly,
obtaining a balanced line that still crosses $A$ but does not contain
points of $F$.
If it contains two points of the same color, we rotate it slightly
about the midpoint between these two points.

Now, $A$ has a \x-end in one half-plane of $\ell$ and a
\ci-end in the other half-plane. Since $\ell$ is balanced, there
must be at least one other segment $B$ that has its
 \x-end and
\ci-end on the opposite sides as $A$.
So, $\ell$ is a chromatic cut.

% Under this assumption it is impossible that $\ell$ crosses \emph{only one} segment from $M$,
% and it is impossible that all other segments of $M$ that $\ell$ crosses
% have \x-ends in the same halfplane as $A$.
% Therefore, $\ell$ crosses another segment $B$
% such that
% \x-ends of $A$ and $B$ are on different sides of $\ell$.
% So, $\ell$ is a chromatic cut.
	\medskip
	
	\noindent ``$\Rightarrow$''
First, let $A$ be {a} segment in $M$,
and let $p$ be an inner point of $A$ that does not belong to any line
determined by two points of~$F$, other than the endpoints of~$A$.
We claim that if there is no balanced line that crosses $A$ at $p$,
then $g(A)$ is the only balanced line through~$p$.

Assume that there is no balanced line that crosses $A$ at $p$.
We use a continuity argument.
Let $m=m_0$ be any directed line that crosses $A$ in $p$.
Rotate $m$ around $p$ counterclockwise
until it makes  a half-turn.
Denote by $m_\alpha$ the line obtained from $m$ after rotation by the angle $\alpha$;
so, we rotate it until we get $m_\pi$.
Let $\varphi$ ($0 < \varphi < \pi$) be the angle such that $m_\varphi$ coincides with $g(A)$.
Assume without loss of generality that the right halfplane bounded by $m$ is dominated by \x;
then the right halfplane bounded by $m_\pi$ is dominated by \ci.
As we rotate $m$, the points of $F$ change the side \emph{one by one},
except at $\alpha=\varphi$.
When one point changes sides, $m_\alpha$ cannot change from
\x-dominance to \ci-dominance without becoming a balanced line.
Therefore, for each $0\leq\alpha<\varphi$, the right side of $m_\alpha$ is dominated by \x,
and for each $\varphi  < \alpha \leq \pi$, the right side of $m_\alpha$ is dominated by \ci.
At $\alpha=\varphi$, exactly two points of different colors change sides.
The only possibility is that the \x-end of $A$ passes from from the right side to the left side
and the \ci-end of $A$ passes from the left side to the right side of the rotated line.
It follows that at this moment the value of
$\# (\bullet) - \# (\circ)$
in the right % open
 halfplane
changes from $1$ to $-1$,
and that $m_\varphi=g(A)$ is a balanced line.

\smallskip

Now, let $\ell$ be a chromatic cut
that crosses $A, B \in M$ so that
the \x-end of $A$ and the \ci-end of $B$ are in the same half-plane bounded by $\ell$.
Denote by $p$ and $q$ the points of intersection of $\ell$ with $A$ and $B$, respectively.
We assume without loss of generality that $p$ and $q$ do not belong to any line determined by points of $F$.

If there is a balanced line that crosses $A$ at $p$,
or a balanced line that crosses $B$ at $q$, %then
 we are done.
By the above claim,
it remains to consider the case when
the only balanced line through $p$ is $g(A)$ and
the only balanced line through $q$ is $g(B)$.
%
% WE DON'T NEED THIS: (gr) ???
% In this case $g(A)$ doesn't cross $B$, {and} $g(B)$ doesn't cross
% $A$.
%
Assume without loss of generality that $\ell$ is horizontal,
 $p$ is left of $q$,
and the \x-end of $A$ is above $\ell$,
see Fig.~\ref{fig:cc} for an illustration.

We start with the line $k=g(A)$, directed upwards, %from below to above),
rotate it clockwise around $p$ until it coincides with $\ell$,
and then continue to rotate it clockwise around $q$ until it coincides with $g(B)$,
directed % from above to below).
down.
%For this changing line
We monitor $\# (\bullet) - \# (\circ)$ on the right side of the line
$k$ as above:
this quantity is $0$ in the initial and the final position.
Just after the initial position it is $-1$,
and just before the final position it is $+1$.
In between,
it makes {only} $\pm1$ jumps,
since the points of $F$ change the side of the rotated line $k$ {one by one}.
It follows that for some intermediate position it is $0$ --- a contradiction.
\end{proof}

\begin{figure}[h]
%$$\includegraphics[width=135mm]{figures/cc}$$ herehere
%\resizebox{160mm}{!}{\input{figures/cc.pdf_t}}
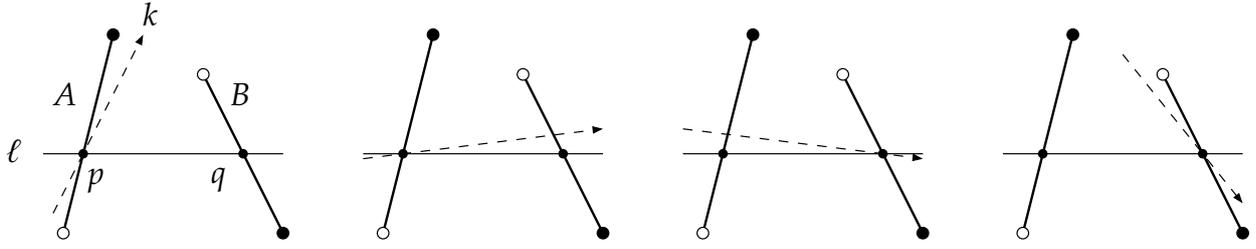
\caption{Illustration to the proof of Lemma~\ref{lem:cutBal}.}
\label{fig:cc}
\end{figure}

In Section~\ref{sec:algo}, we discuss the algorithmic implementation
of this proof.

\begin{corollary}\label{cor:bal}
	Let $M$ be a BR-matching of $F$ with a chromatic cut. Then there is a different matching $M'\neq M$.
\end{corollary}
\begin{proof}
By Lemma \ref{lem:cutBal}, there exists a balanced line $\ell$ crossing
%the interior of
% AAA I removed "the interior" because "crosses" means this (To be stated explicitly?).
%
a segment $A\in M$.
%Match the points to the right of $\ell$ to each other and likewise the points to the left of $\ell$. This matching $M'$ does not use $s$.
We construct matchings on both sides of $\ell$, and denote their union by $M'$.
Then $M'$ is a matching of $F$, and we have $M'\neq M$ since $M'$ does not use $A$.
\end{proof}

%The attentive reader might has noticed that Corollary is already implied by \cite[Lemma 9]{2012arXiv1207.2375A} mentioned in the introduction. We gave an elementary alternative proof.

\begin{remark} As mentioned in the introduction,
Corollary~\ref{cor:bal} follows from
the stronger statement of~\cite[Lemma 9]{abls-bcm-13}:
 the existence of a compatible
matching $M'\ne M$.
We have given a simpler alternative proof.
\end{remark}

\begin{figure}[h]
			\begin{center}
			\includegraphics{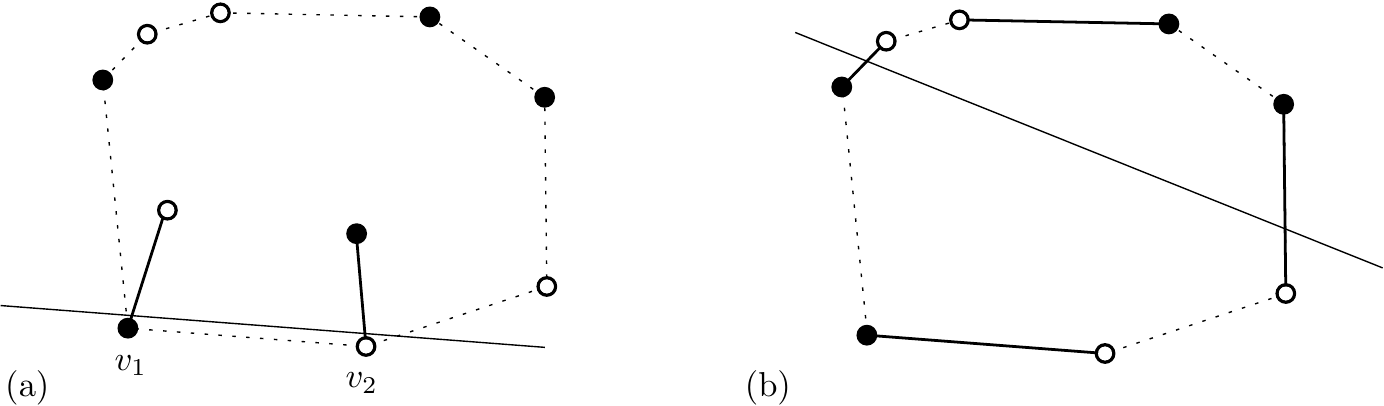}
%				\caption{ a) two vertices on $\partial \CH(F)$ which do not share a segment cause a chromatic cut b) two segments on $\partial \CH(F)$ with the same orientation also cause a chromatic cut c) $v$ can only be matched to one of its neighbors on $\partial \CH(F)$. }
\caption{Illustration to the proof of Lemma~\ref{lem:otherPS_v2}.}
				\label{fig:CutObservation}
			\end{center}
\end{figure}

\begin{lemma}\label{lem:otherPS_v2}
Let $M$ be a BR-matching of $F$ that has no chromatic cut.
Then
\begin{itemize}
\item either all points of $\partial \CH(F)$ have the same color,
\item or the points of $\partial \CH(F)$ %, ordered cicrularly
                                %according to their appearance,
form two color intervals of size at least $2$.
\end{itemize}
In the latter case, the two boundary segments connecting points of different color
necessarily belong to $M$.
\end{lemma}

\begin{proof}
Assume that $\partial \CH(F)$ has points of both colors.

If $v_1$ and $v_2$ are two neighboring points on $\partial \CH(F)$ with different colors,
then they are matched by a segment of $M$.
Indeed, let $\ell'$ be the line through $v_1$ and $v_2$.
If $v_1$ and $v_2$ are not matched by a segment of $M$, then
each of them is an endpoint of some segment of $M$.
When we shift $\ell'$ slightly so that it crosses these two segments,
a chromatic cut is obtained, see Fig.~\ref{fig:CutObservation}~(a).

Therefore, if the points of $\partial \CH(F)$ form more than two color intervals,
then at least four segments of $M$ have both ends on $\partial \CH(F)$.
At least two among them have the \x-end before the \ci-end, with respect to their circular order.
Any line that crosses these two segments will be then a chromatic cut, see Fig.~\ref{fig:CutObservation}~(b).

Thus, we have exactly two color intervals.
If one of them consists of one point, then this point has two neighbors of another color of $\partial \CH(F)$.
As observed above, this point must be matched by $M$ to both of them, which is clearly impossible.
\end{proof}

We recall the notation from the introduction:
a \emph{linear matching} is a BR-matching without a chromatic cut
such that $\partial \CH(F)$ consists of exactly two color intervals, both of size at least $2$;
a \emph{circular matching} is a BR-matching without a chromatic cut
such that all points of $\partial \CH(F)$ have the same color.
So, we proved now that if a BR-matching has no chromatic cut, then
it necessarily belongs to one of these types.
In the next sections we study these types in more detail.

%\begin{figure}[h]
%			\begin{center}
%			\includegraphics[width = 0.8\textwidth]{figures/unordered}
%				\caption{These are the only configurations up to reflection and exchange of all colors that two segments are not comparable. In this case they have a chromatic cut.}
%				\label{fig:unordered}
%			\end{center}
%\end{figure}

%type 2 => ordered  %%%%%%%%%%%%%%%%%%%%%%%%%%%%%%%%%%%%%%%%%%%%%%%%%%%%%%%%

% GGR Adding \section
%%%%%%%%%%%%%%%%%%%%%%%%%%%%%%%%%%%%%%%%%%%%%%%%%%%%%%%%%%%%%%%%%
\section{A Sidedness Relation between Segments}
\label{sec:order}

\begin{definition}\label{def:orderM}
	For two segments $A, B$,
 we define the \emph{sidedness relation} \ord \ 
%(...%\textipa{["traI\ae Ng@l]}
%)  
 as follows:
	$ A \triangleleft B $ if
	$B$ lies strictly right of $g(A)$ and
	$A$ lies strictly left of $g(B)$.
\end{definition}

%We suggest to pronounce \ord \ as triangel.

\begin{lemma}\label{lemma-total-order}
Let $M$ be a BR-matching.
$M$ has no chromatic cut iff
the sidedness relation $\triangleleft $
is a total relation, i.\,e.,
for any two segments $A, B\in M$, $A\neq B$,
 we have
	$ A \triangleleft B $
or	$ B \triangleleft A $.
\end{lemma}
\begin{proof}
If two segments have a chromatic cut, then their
 supporting lines must
intersect as in Fig.~\ref{fig:allConfig}~(a) or~(b),
and the segments are not comparable by~$\triangleleft$;
otherwise, the supporting lines
are parallel or intersect in the outer rays of the same color,
and then the segments are comparable.
\end{proof}

Note that the relation $\triangleleft$ is asymmetric by definition:
we never have
	$ A \triangleleft B $
and	$ B \triangleleft A $.
Moreover, if $M$ has no chromatic cut, then, in order to prove $A \triangleleft B$,
it suffices to prove only one condition from the definition.

\begin{lemma}\label{lem:half_triangle}
Let $M$ be a BR-matching without chromatic cut; $A, B  \in M$  $(A\neq B)$. If
$B$ lies right  of $g(A)$
\emph{\textbf{or}} $A$ lies  left of $g(B)$ %\marginpar{AA: Clear enough that only one cond. is meant?},
then $A \triangleleft B$.
\end{lemma}
\begin{proof}
If $M$ has no chromatic cut by Lemma \ref{lemma-total-order} either $A \ord B$ or $B \ord A$. If we know
one of the above conditions, $B \ord A$ is ruled out.
%Assume $B$ is contained in the right halfplane bounded by $g(A)$.
%Since $M$ has no chromatic cut, $g(A)$ and $g(B)$ are either parallel,
%or intersect in the outer rays of the same color.
%In both cases, $A$ is contained in the right halfplane bounded by $g(B)$. \marginpar{Clear?}
\end{proof}

%%%%%%%%%%%%%%%%%%%%%%%%%%%%%%%%%%%%%%%%%%%%%%%%%%%%%%%%%%%%%%%%%

%%%%%%%%%%%%%%%%%%%%%%%%%%%%%%%%%%%%%%%%%%%%%%%%%%%%%%%%%%%%%%%%%
\section{Quasi-Parallel, or Linear, Matchings}
\label{sec:type_l}
%%%%%%%%%%%%%%%%%%%%%%%%%%%%%%%%%%%%%%%%%%%%%%%%%%%%%%%%%%%%%%%%%

In this section we give several characterizations of linear matchings
and
prove that such matchings are unique for their point sets.

%We start with two definitions.

\begin{lemma}\label{lem:min_max}
If $M$ is a linear matching,
then it has a minimum and a maximum element with respect to $\triangleleft$.
\end{lemma}
\begin{proof}
By Lemma~\ref{lem:otherPS_v2}, the two boundary segments connecting points of different color belong to $M$.
For one of them (to be denoted by $A_1$), all other segments of $M$ belong to the right half-plane bounded by $g(A_1)$;
for the second (to be denoted by $A_n$), all other segments of $M$ belong to the left half-plane bounded by $g(A_n)$.
Since $M$ has no chromatic cut, it follows from Lemma~\ref{lem:half_triangle}
that $A_1$ is the minimum, and $A_n$ is the maximum element of $M$ with respect to $\triangleleft$.
\end{proof}

\begin{definition}\label{def:quasiParallel}
A BR-matching $M$ is called \emph{quasi-parallel}
if there exists a directed \emph{reference line} $\ell$ such that the following conditions hold:
\begin{itemize}
	\item
[\rm(i)]
 No segment is perpendicular to $\ell$.
%	\item The projected \x-point of any segment is further to the ``right'' on $\ell$ then the \ci-point of the same segment. See Figure~\ref{fig:quasiPar}, where $\ell$ is horizontal.
\item
[\rm(ii)] For any $A \in M$, the direction of its projection on $\ell$ (as usual, from \ci\ to \x) coincides with the direction of $\ell$.
%	\item Consider the projection of the segments $a,b \in M$ and the intersection of the carrying lines $g(a)$ and $g(b)$ on $\ell$. We require that the projected intersection lies not on nor between these two projected segments.
\item
[\rm(iii)] For any non-parallel $A, B \in M$, the projection of the
intersection point of $g(A)$ and $g(B)$ on $\ell$
 lies outside the convex hull of the projections of $A$ and $B$ on $\ell$.
\end{itemize}
\end{definition}

The notion of quasi-parallel segments was introduced by %G\"{u}nter
 Rote~\cite{RoteJan1992, RoteDis1988}
as a generalization of parallel segments,
in the context of a dynamic programming algorithm for some instances of the traveling salesman problem.
In that work, the segments were uncolored;
thus, our definition is a ``colored'' version of the original one.
Fig.~\ref{fig:quasiPar} shows an example of quasi-parallel matching, with horizontal $\ell$.
%The segments G\"{u}nter Rote considered were uncolored. It is easy to go from one definition to the other.
%We have chosen this form for convenience.

\begin{figure}[h]
			\begin{center}
			\includegraphics[width = \textwidth]{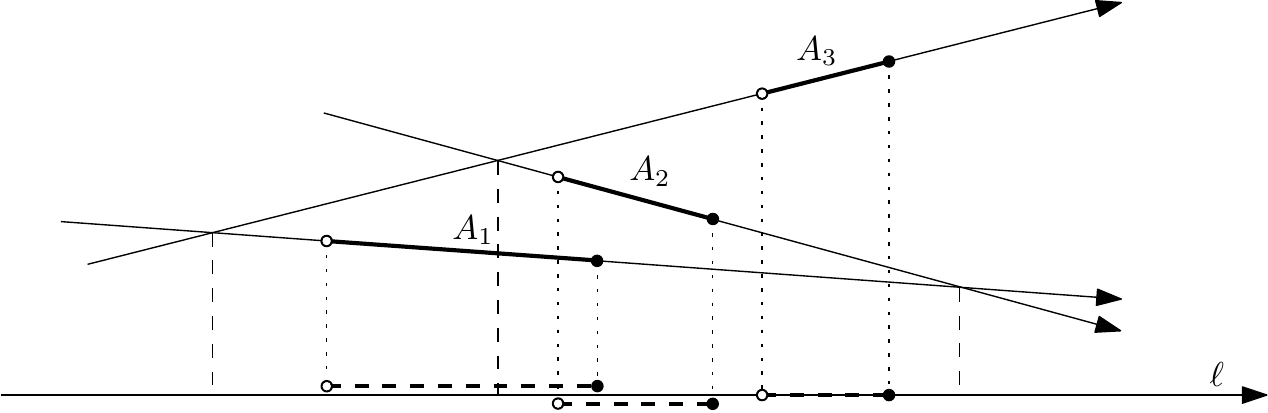}
				\caption{A quasi-parallel matching.}
				\label{fig:quasiPar}
			\end{center}
\end{figure}

%Refer now to
 Fig.~\ref{fig:allConfig}.
 shows three ``patterns'' --- configurations of two or three segments
with respect to the intersection pattern of their supporting lines
(thus, they can be expressed in terms of order types).
The patterns are considered up to exchanging the colors and reflection.
A pair of antiparallel segments is considered as a special case of the pattern~(b).

\begin{lemma}\label{lem:typeL}
Let $M$ be a BR-matching of a bichromatic $(n+n)$ set $F$.
Then the following conditions are equivalent:  
\begin{enumerate}
\item $M$ is a linear matching.
\item The relation $\triangleleft$ in $M$ is a strict linear order.
\item $M$ has no patterns of the three kinds in Fig.~\ref{fig:allConfig}.
\item $M$ is a quasi-parallel matching.
\end{enumerate}
\end{lemma}

\begin{proof}
``$1 \Rightarrow 2$'' By definition, the relation $\triangleleft$ is asymmetric, 
and according to Lemma~\ref{lemma-total-order}, it is total.

It remains to prove transitivity. 
As a linear matching,  $M$ has a minimum
$A_1$ and a maximum $A_n$ with respect to \ord \  by Lemma~\ref{lem:min_max}. We define inductively
$A_2,\dots , A_{n-1}$ as follows. Assume $A_1,\dots ,A_{i-1}$ are already defined. Remove $A_1,\dots , A_{i-1}$ from $M$. Then the new matching is again of linear type. It has still no chromatic cut
and still both colors on the boundary of the convex hull because $A_n$ belongs to it. Denote the new minimum element by $A_i$ and repeat.

Note that $A_j$ lies to the right of $g(A_i)$, $\forall i<j$, by construction.
Thus $A_i \ord A_j$ , $\forall i<j$, by Lemma~\ref{lem:half_triangle}.
This implies transitivity of \ord .

\smallskip

``$2 \Rightarrow 3$'' It is easy to check that none of the
configurations in Fig.~\ref{fig:allConfig} is ordered linearly
by~\ord.

\smallskip

``$3 \Rightarrow 4$''
In this proof, we follow the idea from~\cite{RoteDis1988,RoteJan1992}.
%The following Lemma uses an elegant idea from Rote~\cite{RoteDis1988} to find a line $\ell$ given that any three %segments are quasi-parallel.
As a preparation, one can establish by case distinction that any two or three segments
which contain none of the patterns from Fig.~\ref{fig:allConfig} are quasi-parallel.

%If a BR-matching $M$ is quasi-parallel then there exists a line $\ell$.
%We fix an arbitrary segment $g_0$ and any coloring of the endpoints. The condition, that chromatic cuts are forbidden fixes a coloring of all other segments. Any other two segments $a,b$ have no chromatic cut, because the segments $g_0,a,b$ are quasi-parallel.
Now, let $M$ be a BR-matching without the forbidden patterns in Fig.~\ref{fig:allConfig}.
For each $A\in M$, let $a(A)$ be the arc on the circle of directions corresponding to positive directions of lines $m$
such that the angle between $A$ and $m$ is acute.
(These are the lines that can play the role of a reference line $\ell$ in the definition of quasi-parallel matching, with respect to $A$.)
Each $a(A)$ is an open half-circle, see Fig.~\ref{fig:halfcircle}.
\begin{figure}[h]
			\begin{center}
			\includegraphics[width = 35mm]{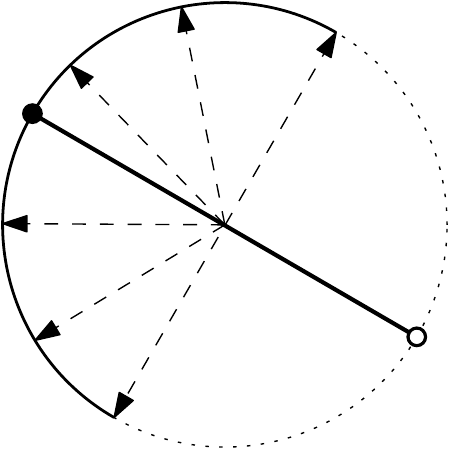}
%				\caption{Given a colored segment the possible orientations of $\ell$ form an open half-circle.}
				\caption{The open
arc $a(A)$ for a matching segment $A$, used in the proof of Lemma~\ref{lem:typeL}, $3 \Rightarrow 4$.}
				\label{fig:halfcircle}
			\end{center}
\end{figure}

Fix a segment $S$ of $M$.
For any segments $A\in M$.
 $\{S, A\}$ is a quasi-parallel matching,
and hence the intersection of the corresponding arcs
$a(S)\cap a(A)$ is a non-empty sub-interval of $a(S)$,
which we denote by $a'(A)$.
Now, for any two segments $A, B \in M$,
$\{S, A, B\}$ is a quasi-parallel matching,
and hence the intersection of the corresponding arcs is non-empty.
In other words, $a'(A)\cap a'(B)\ne \emptyset$.
We can apply Helly's Theorem to the intervals $a'(A)$ (considering
them as one-dimensional subintervals of $a(S)$) and conclude that
there exists a direction in the intersection of the arcs corresponding
to all segments of $M$.
A line $\ell$ in this direction will fulfill conditions~(i) and~(ii) of the
definition of quasi-parallel matching.
Finally, the absence of forbidden patterns implies that
 condition~(iii) %in the definition of quasi-parallel matchings
 is satisfied as well.
%
% Let $A, B$ be any two segments in $M$.
% Since, by the explained above, $\{S, A, B\}$ is a quasi-parallel matching,
% the intersection of the corresponding arcs is non-empty.
% In other words, $a(A)$ and $a(B)$ intersect, and their intersection belongs to $a(S)$.
% Therefore, by Helly's Theorem, the intersection of the arcs corresponding to all  segments of $M$, is non-empty.
% Therefore, it is possible to find a line $\ell$ as in the definition of quasi-parallel matching.

% Finally, once such a line is found,
% the absense of forbidden patterns implies that
% the last condition in the definition of quasi-parallel matchings is satisfied as well.

\smallskip

``$4 \Rightarrow 1$''
 Condition~(iii) in the definition of quasi-parallel matchings implies that for any $A, B \in M$, $A \neq B$,
the lines $g(A)$ and $g(B)$ are either parallel, or the outer rays of the same color cross.
It follows from Lemma~\ref{lem:chCCut} that there is no chromatic cut.

The highest point and the lowest point of $F$,
with respect to $\ell$, belong to the boundary of the convex hull and have different colors.
%The extremal points on $\ell$ have different color and are both on the boundary of the convex hull.
%Thus the BR-matching is not of type C.
%and since we have no chromatic cut. Lemma \ref{lem:otherPS} asserts that the resulting BR-matching is of type $2$.
Therefore, $M$ is of linear type.
\end{proof}

\begin{remark} The characterization of quasi-parallel matchings by
forbidden patterns was found earlier by
Rote in~\cite{RoteDis1988}.
We have fewer forbidden patterns,
because we deal with colored segments.
(In the journal version~\cite{RoteJan1992} of this result,
the list of patterns was incomplete: the pattern corresponding to
Fig.~\ref{fig:allConfig}~(c) had been overlooked.)
\end{remark}

Lemma~\ref{lem:typeL} proves the equivalence of conditions
\ref{type_l},
\ref{order},
\ref{forbidden},
and \ref{qp} in Theorem~\ref{thm:main}.
Condition~\ref{order} justifies the term ``matching of linear type''.
Now we prove that they imply the uniqueness of $M$.

\begin{theorem}\label{thm:type2unique}
	Let $M$ be a linear matching on the point set $F$.
	Then $M$ is the only matching of~$F$.
\end{theorem}

\begin{proof}
By Lemma~\ref{lem:typeL}, the matching $M$ is quasi-parallel,
%In particular, they are ordered linearly by $\triangleleft$.
%
%Let $\ell$ be a reference line from the definition of quasi-parallel matching.
with reference line~$\ell$.
We assume without loss of generality that $\ell$ is vertical.

Assume for contradiction that another matching $M'$ exists.
(In the figures below, the segments of $M$ are denoted by solid lines, and the segments of $M'$ by dashed lines.)
%Let $A$ be the ``leftmost'' segment of $M$ that does not belong to $M'$. This means that to the left of $A$, the matchings $M$ and $M'$ coincide,
%and so we assume WLOG that $A$ is the leftmost segment of $M$.
The symmetric difference of $M$ and $M'$ is the union of alternating
cycles.
We now claim that \emph{an alternating cycle must intersect itself}.

Consider the alternating cycle $\Pi = p_1 q_1 p_2 q_2 p_3 q_3 \dots
p_1$ that consists of segments of $p_iq_i\in M$ and $q_ip_{i+1}\in
M'$.
We assume that $p_i$ are \ci-vertices and $q_i$ are \x-vertices.
%
% We construct an \emph{alternating path for $M$} as follows.
% Let $A$ be any segment of $M$ that does not belong to $M'$.
% Let $p_1$ be the \ci -end of $A$, and let $q_1$ be the \x -end of $A$.
% For $i > 1$, let $p_i$ be the \ci -point connected to $q_{i-1}$ by a segment from $M'$
% and let $q_i$ be the \x -point connected to $p_{i}$ by a segment from $M$:
% %it is clear that eventually we shall return to $P_1$,
% we apply this labeling until we eventually return to $p_1$.
% In this way we obtain a closed path $\Pi = p_1 q_1 p_2 q_2 p_3 q_3 \dots p_1$ that consists, interchangeably, of segments of $M$ and $M'$.
% (Notice that a matching may have several disjoint altenating paths. The arguments below are valid for \emph{any} alternating path.)
%
%
Let $B$ be the minimum (with respect to $\triangleleft$) segment and let $C$ be the maximum segment of $M$ that belongs to $\Pi$.
Then no points of $\Pi$ lie left of $g(B)$ or right of $g(C)$.
Since for both $B$ and $C$ the \x-end is higher than the \ci-end, the
path $\Pi$ must cross itself at least once, establishing the claim,
see Fig.~\ref{fig:path}.
%\marginpar{the lines g(A) are too thin. \texttt{AA: OK now?}}

\begin{figure}[h]
$$\includegraphics{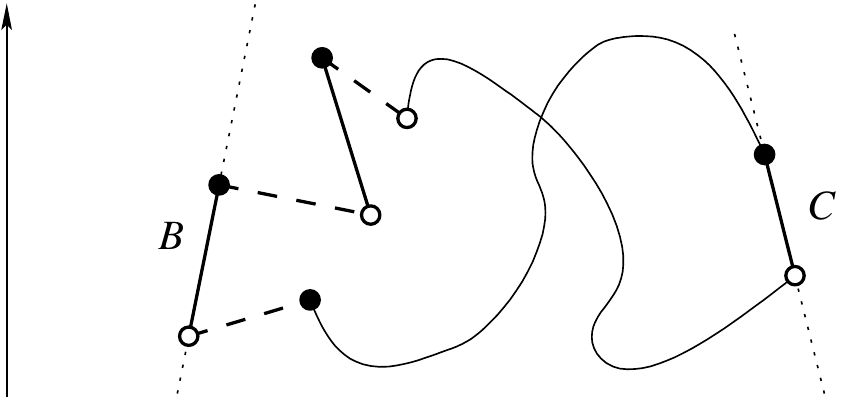}$$
\caption{Illustration for the proof of Theorem~\ref{thm:type2unique}: an alternating path for $M$ crosses itself.}
\label{fig:path}
\end{figure}

We now traverse the path $\Pi$, starting from $p_1 q_1 p_2 q_2\dots$, until it crosses itself {for the first time}, say, in a point $r$.
There can be no crossing $r$ between two segments of $M$ or two segments of
$M'$.
Hence,
%Without loss of generality,
the first occurrence of $r$ on $\Pi$ is on a segment $p_i q_i$ of $M$,
and the second is on a segment $q_jp_{j+1}$ of $M'$,
or vice versa. We consider only the first case, the other being similar.
%(the reader should verify that the case when the first occurrence of $R$ is on a segment of $M'$, and the second is on a segment of $M$, is similar).
In this case, we consider the matching $N$ that consists of segments
$r q_i, p_{i+1} q_{i+1}, p_{i+2} q_{i+2}, \dots p_{j}q_{j}$
(that is, $N$ consists of the segments of $M$ that occur on $\Pi$ between the two times that it visits $r$,
and the part of segment of $M$ that contains $r$).
It is clear that $N$ is also quasi-parallel.

\begin{figure}[h]
$$\includegraphics{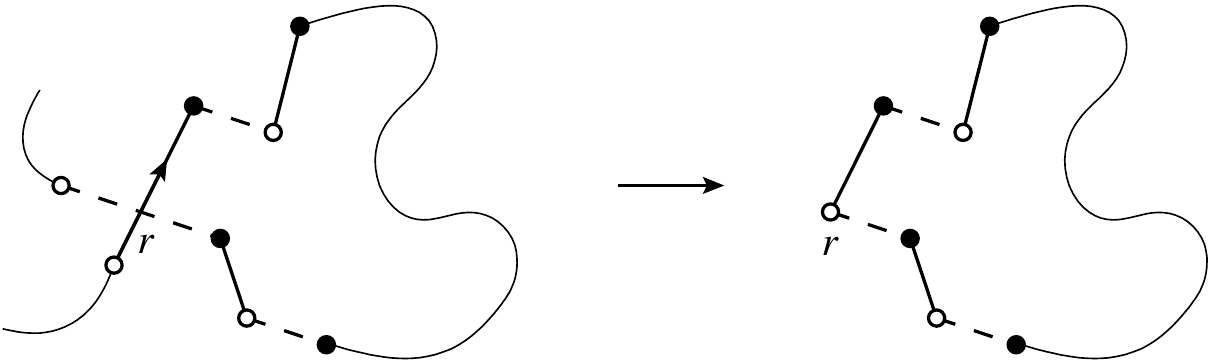}$$
\caption{Illustration to the proof of Theorem~\ref{thm:type2unique}: an alternating path for $N$.}
\label{fig:path2}
\end{figure}

The closed path $r q_i p_{i+1} q_{i+1} p_{i+2} q_{i+2} \dots p_{j} q_{j} r$ is an alternating path for $N$.
By the choice of~$r$, this path does not intersect itself, see
Fig.~\ref{fig:path2},
which contradicts the claim proved above that an alternating path of a quasi-parallel matching always intersects itself.
%$$\includegraphics[width=180pt]{path3}$$
%$$\includegraphics[width=180pt]{path4}$$
\end{proof}

The proof we have just given established the uniqueness of $M$
directly. A weaker version of Theorem~\ref{thm:type2unique} was known
before~\cite[Lemma 2]{RoteDis1988,RoteJan1992}: for a linear matching $M$, there is no other \emph{compatible} matching $M'\ne M$.
This % statement
implies Theorem \ref{thm:type2unique} by the fact that the compatible
matching graph is connected \cite{abls-bcm-13}.

\medskip

The proof of Theorem~\ref{thm:type2unique} tells us that a closed alternating path cannot exist.
In contrast, it is always possible to construct two \textit{open} alternating paths from the minimum to the maximum element of $M$:

\begin{lemma}\label{lem:altpath}
	Let $M$ be a linear matching. 
	Then there exist two alternating paths containing all segments of $M$, appearing according to the order \ord .
\end{lemma}
\begin{proof}
Let $A_1, \dots, A_n$ be the segments of $M$, ordered by \ord.
We proceed by induction.
%	It is obviously true when $M$ is only one segment. 
Let $R_k$ be a path from $A_1$ to $A_k$ in which the segments of $M$ appear according to \ord. 
We obtain $R_{k+1}$ by taking $R_k$ and adding a color-conforming segment from $A_k$ to $A_{k+1}$. 
This is possible because there is no other segment of $M$ between $A_k$ and $A_{k+1}$. 
The color of the starting point can be chosen and thus we have two such paths, see Fig.~\ref{fig:open_alt_path}.
\end{proof}

\begin{figure}[h]
\centering
\includegraphics{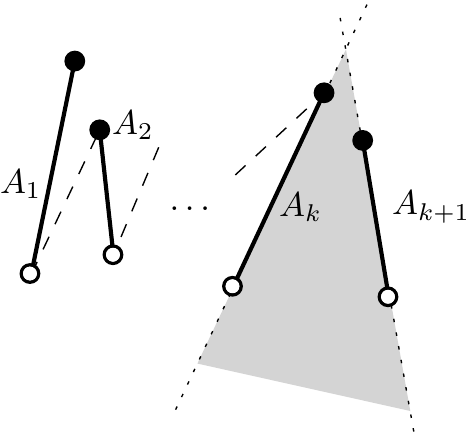}
\caption{Illustration to the proof of Lemma~\ref{lem:altpath}.}
\label{fig:open_alt_path}
\end{figure}

\paragraph{Proof of Theorem \ref{thm:type2unique} by the Fishnet Lemma.}

We will give another proof, which
captures the intuition that one gets when drawing an
alternating path and trying to close it.
Indeed,
%Assume for contradiction
%that there is an additional noncrossing matching $M'$.
%Then there must be an alternating cycle in the union of $M$ and $M'$.
when one starts to construct an alternating path
as in the first proof,
one quickly gets the feeling of being stuck:
even though it is permitted that segments of
$M$ and $M'$ may cross, one cannot close the path because one is
forced in one direction.
This feeling can be made precise with the following \emph{Fishnet Lemma}.
We will apply it only to polygonal curves, but we formulate it for
arbitrary curves, see Fig.~\ref{snakes}.

Consider a set $V=\{v_1,\dots,v_m\}$ of pairwise noncrossing unbounded Jordan
curves (``ropes'').
They partition the plane into $m+1$ connected regions. We assume that
they are numbered in such a way that in going from $v_i$ to $v_j$
($j>i+1$), one has to cross $v_{i+1},v_{i+2},\dots,v_{j-1}$.
These curves will be called the
% \emph{black} curves or
 \emph{vertical} curves.
In the illustrations, they will be black, and
 we think of them as % ``vertical'' curves
 numbered from left to right.

Consider another set $G=\{g_1,\dots,g_n\}$ of pairwise noncrossing Jordan
arcs,
 % (the \emph{green} arcs)
 % (the \emph{horizontal} arcs),
called the \emph{horizontal} arcs and drawn in {green},
such that every curve $g_k$ has its
endpoints on two different vertical curves
 $v_i$ and $v_j$
($j>i$), has exactly one intersection point with each vertical curve
 $v_i,v_{i+1},v_{i+2},\dots,v_j$, and no intersection with the other
 curves.
See Fig.~\ref{snakes}~(a) for an example.
We say that the curves $V\cup G$ form a \emph{partial
(combinatorial)
grid.}

\begin{figure}[ht]
  \centering
  \includegraphics{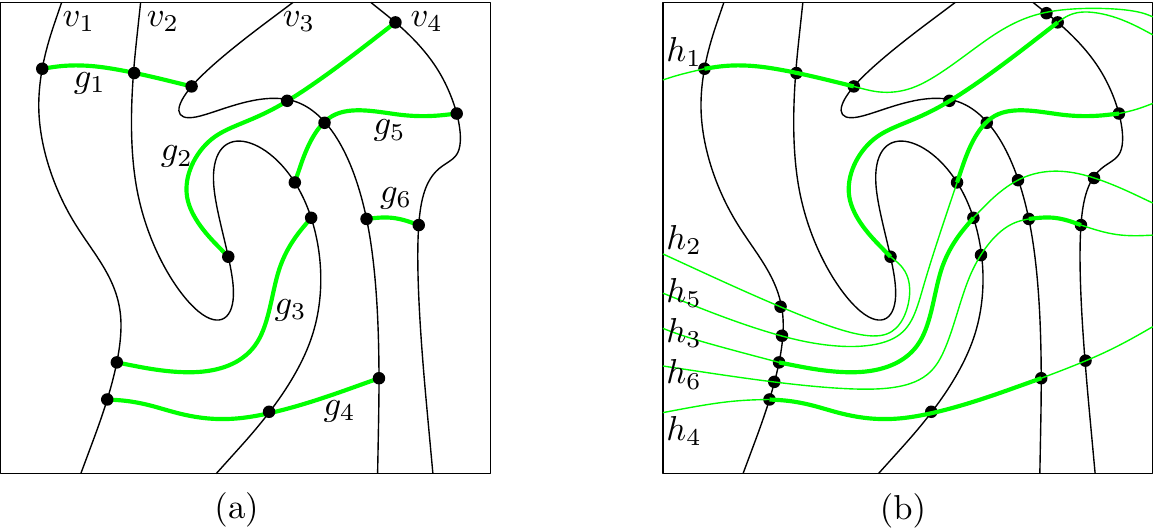}
  \caption{(a) A partial grid. (b) Extension to a full grid of ropes.}
  \label{snakes}
\end{figure}

\begin{lemma}[The Fishnet Lemma]\label{lem:fishnet}
  The horizontal arcs $g_k$ of a
partial
combinatorial
grid $V\cup G$ can be extended to pairwise noncrossing unbounded Jordan arcs $h_k$ in such a way
that the curves
 $H=\{h_1,\dots,h_n\}$ together with $V$ form a full
combinatorial
grid $V\cup H$: each horizontal curve $h_k$ crosses each vertical
curve $v_i$ exactly once.
See Fig.~\ref{snakes}~(b).
\end{lemma}

\begin{proof}
  This is an easy construction, which incrementally grows the
  horizontal segments.

  The bounded faces of the given curve arrangement $V\cup G$ are
  quadrilaterals: they are bounded by two consecutive vertical curves
  and two horizontal curves.

 The bounded faces of the
  desired final curve arrangement  $V\cup H$ are also such
  quadrilaterals,
with the additional property that they have no extra vertices on their
boundary besides the four corner intersections.
In $V\cup G$, such extra vertices arise as the endpoints of the
segments $g_k$.

Let us take such a bounded face, between two vertical curves $v_i$ and
$v_{i+1}$, with an endpoint of $g_k$ on one of its vertical sides, see
Fig.~\ref{snakes-add}~(a)--(b).  We can easily extend $g_k$ to some point
on the opposite vertical side, chosen to be distinct from all other
endpoints, splitting the face into two and creating a new intersection
point.  (The existence of such an extension follows from the
Jordan--Schoenflies Theorem, by which the bounded face is homeomorphic
to a disc.)
\begin{figure}[ht]
  \centering
  \includegraphics{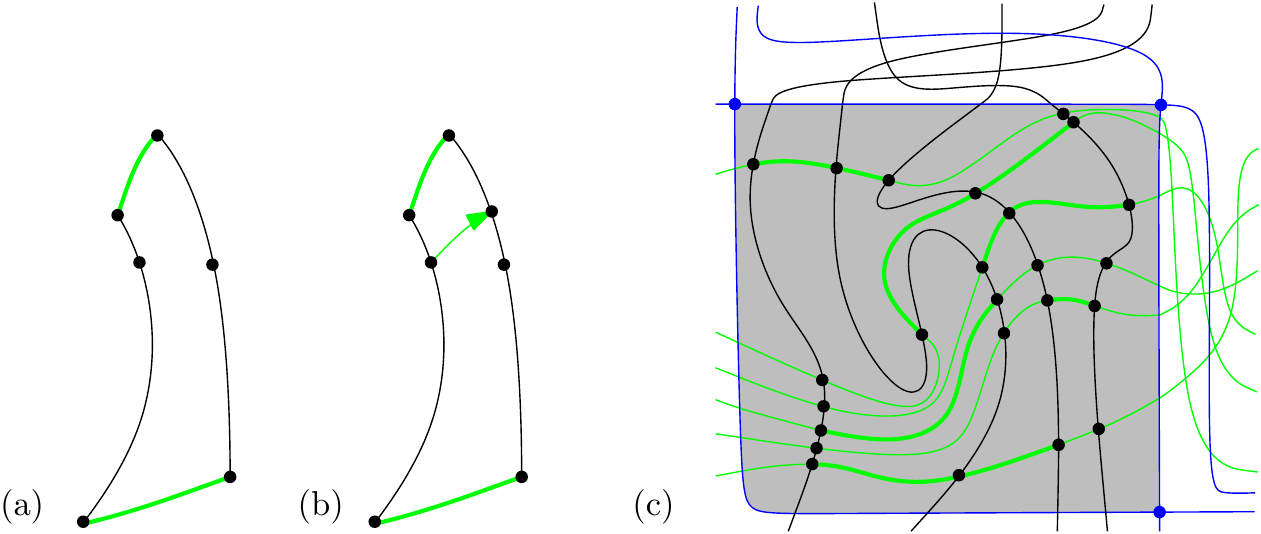}
  \caption{(a) A quadrilateral face with extra vertices;
the shaded face from Fig.~\ref{snakes}~(a).
 (b)~Adding an
    edge.
(c)~Embedding the grid into a pseudoline arrangement.}
  \label{snakes-add}
\end{figure}
An unbounded face between two successive vertical curves $v_i$ and
$v_{i+1}$ that has an extra vertex on a vertical side can be treated
similarly.

We continue the above extension procedure as long as possible.
Since we are adding new intersection points, but no two curves can
intersect twice, this must terminate. Now we are almost done:
each horizontal curve extends from $v_1$ to $v_m$ and
crosses each vertical curve exactly once.
Now we just extend the horizontal curves to infinity,
left of $v_1$, and right of $v_m$, without crossings.
%[ Do we have to prove that this can be done?? ]
% NO: it follows from repeated application of the Jordan--Schoenflies Theorem
\end{proof}

This lemma can be interpreted in the context of pseudoline
arrangements.
In an  arrangement of pseudolines, each pseudoline is an unbounded
Jordan curve, and every pair of pseudolines has to cross
\emph{exactly} once. The grid construction can be embedded in a true
pseudoline arrangement, see Fig.~\ref{snakes-add}~(c): simply enclose all crossings in a bounded region
formed by three new (blue) pseudolines and let the crossings
between vertical lines and between horizontal lines occur outside this region.

We return to the proof of Theorem~\ref{thm:type2unique}.
 %[ PROBABLY THIS IS AN INSTANCE OF some
%well-known theorem in ps-line arrangements. (Ask STEFAN?) ]

%To apply the Fishnet Lemma
%for proving Theorem~\ref{thm:type2unique}, suppose
%that there is an additional matching $M'$. We look at an alternating
%path
%between two matchings $M$ and $M'$.
%The lines $g(e)$ for the segments $e\in M$ must be uncrossed wherever
%they meet in order
%to get a family of vertical curves for a combinatorial grid.
%[EXPAND THIS AND/OR DRAW A PICTURE]
%
%... etc. always go in one dirsection ... cannot close ...
%%%%%%%%%%%%%%%%%%%%%%%%%%%%%%%%%%%%%%%%%%%%%%%%%%%%%%%%%%%%%%%%%%%%%%%%%%%%%%%%%%%%%%%%%%%%%%%%%
%%%%%%%%%%%%%%%%%%%%%%%%%%%%%%%%%%%%%%%%%%%%%%%%%%%%%%%%%%%%%%%%%%%%%%%%%%%%%%%%%%%%%%%%%%%%%%%%%

\begin{figure}[htb]
  \centering
\includegraphics[]{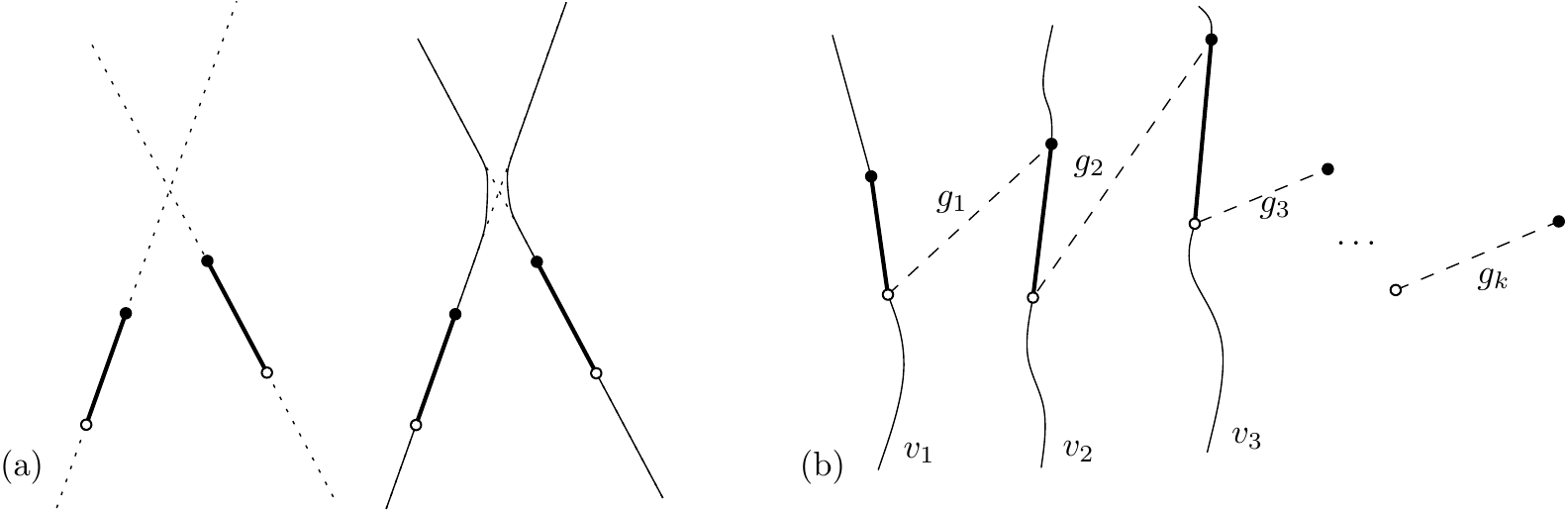}
  \caption{Applying the Fishnet Lemma.}
  \label{fig:ApplyFishnet}
\end{figure}

\begin{proof}
  Given a quasi-parallel matching $M$ we construct a set 
of Jordan curves $V$ as in Lemma~\ref{lem:fishnet}
by considering the
line arrangement formed by the segments $s_1 \ord \dots \ord s_n $ with the corresponding 
lines $g(s_1),\dots,g(s_n)$. We construct curve $v_i$ by going along $g(s_i$). At each intersection, the curves switch from one line to the other, and after a slight deformation in
the vicinity of the intersections, they become non-crossing, see
Fig.~\ref{fig:ApplyFishnet}~(a). These crossings lie outside the parts
of the lines where the segments lie; therefore the switching have no
influence on the left-to-right order of the segments $s_i$;

Now assume there is another matching $M'$;
$M$ and $M'$ form at least one alternating cycle.
Let $G = \{ g_1,\dots,g_k \}$ be the segments of $M'$
on such a cycle in the 
order in which they are traversed.
%  in the 
% order when one would traverse this cycle.
 $V$ and $G$ satisfy the condition of the Fishnet Lemma
and thus can be extended to a full combinatorial grid.
% Observe $g_1,\dots ,g_k$ also must be in this order,
% because (without loss of generality) $h_1$ is above $h_2$, $h_2$ above $h_3$, and so on.
Assume without loss of generality that $g_1$ is above $g_2$ on the
common incident edge of $M$. Then $g_2$ must also lie above $g_3$, and so on.
It follows that the extended horizontal curves $h_1,\dots ,h_k$ also
must be in this order, and
 $g_k$ would lie below $g_1$.
So they cannot be connected to the same segment in $M$.
$\lightning$
\end{proof}

%\texttt{END SNAKE}

%%%%%%%%%%%%%%%%%%%%%%%%%%%%%%%%%%%%%%%%%%%%%%%%%%%%%%%%%%%%%%%%%
\section{Circular Matchings}
\label{sec:type_c}
%%%%%%%%%%%%%%%%%%%%%%%%%%%%%%%%%%%%%%%%%%%%%%%%%%%%%%%%%%%%%%%%%

In this section we study circular matchings in more detail.
%
%?
Recall that such a matching is
a BR-matching without a chromatic cut
for which all points on the convex hull have the same color.
We assume without loss of generality that
this color is % all points of $\partial \CH(F)$ are 
\x.

%Lemma~\ref{lemma-total-order} established that
% the relation $\triangleleft$ in such a matching
%is a total relation.

 We prove that if $M$ is of circular type, then its point set has at least two other matchings.
Moreover, we show that for a circular matching,
 the relation $\triangleleft$ induces a \emph{circular order}
(this will justify the term  ``matching of circular type''),
and describe such matchings in terms of forbidden patterns.

%\emph{Throughout this section, it is assumed that $M$ is a matching of the circular type on a point set $F$.} \marginpar{cancel this}

%%%%%%%%%%%%%%%%%%%%%%%%%%%%%%%%%%%%%%%%%%%%%%%%%%%%%%%%%%%%%%%%%%%%%%%%%%%%%%%%%%%%%%%%%%%%%%%%%

%We will show in Theorem~\ref{thm:typeCNotUnique} that type~$C$ BR-matchings have two disjoint compatible BR-matchings. 

%It is interesting that you could think of a type~C BR-matching as the union of three quasi-parallel matchings. We will use that feature. So let us first find out more about type L matchings.

%\begin{lemma}\label{cor:xPoly}
	%Let $M$ be a BR-matching of $F$ and $Q$ a simple polygon bounded by \x-rays and $s\in M$ a segment in
	%the interior of $Q$. Then $M$ admits a chromatic cut.
%\end{lemma}
%\begin{proof}
	%$s$ admits a \ci-ray, which crosses $Q$. At the crossing a \ci-ray crosses a \x-ray.
%\end{proof}

\begin{lemma}\label{lem:forbiddenConfig}
	A BR-matching $M$ is of circular type if and only if it has no patterns (a) and (b) from
	Fig.~\ref{fig:allConfig}, and has at least one pattern (c) (a ``$3$-star'').
\end{lemma}
\begin{proof}
%A matching of type C is by definition not of type L. Thus there must be one of the forbidden 	subconfigurations of Figure~\ref{fig:allConfig}. %Since $M$ has no chromatic cut, there must be three segments in $M$ as in the Figure~\ref{fig:allConfig}~(c).
We saw in Lemma~\ref{lem:chCCut} that a BR-matching has no chromatic cut if and only if it avoids the patterns (a) and (b).
By Lemma~\ref{lem:otherPS_v2}, a BR-matching without chromatic cut is either of type L or of type C.
By Lemma~\ref{lem:typeL}, a BR-matching is of type L if and only if it avoids (a), (b) and (c).
Therefore, a BR-matching is of type C if and only if it avoids (a) and (b), but contains (c).
\end{proof}

%%In particular, it follows that any matching of circular type is a $3$-star.

\begin{theorem}\label{thm:typeCNotUnique}
Let $M$ be a matching of circular type on the point set $F$.
	Then there exist
(at least) two disjoint BR-matchings on $F$, compatible to $M$.
\end{theorem}
\begin{proof}
According to Lemma~\ref{lem:forbiddenConfig} there are segments in the $3$-star configuration as in Fig.~\ref{fig:allConfig}. 
They partition
% AA: I am not sure that "partition" is the best word, but I think "subdivision" means "further division".
%
the plane into three convex regions $Q_1$, $Q_2$ and $Q_3$ and a triangle as in Fig.~\ref{fig:regions}~(a).
The triangle is bounded (without loss of generality) by three \ci-rays, 
and it is empty of segments: indeed, any segment of $M$ inside the triangle would emit a \x-ray, 
which would cross a \ci-ray --- a contradiction to the fact that $M$ has no chromatic cut, see Lemma~\ref{lem:chCCut}.

\begin{figure}[h]
\centering
\includegraphics{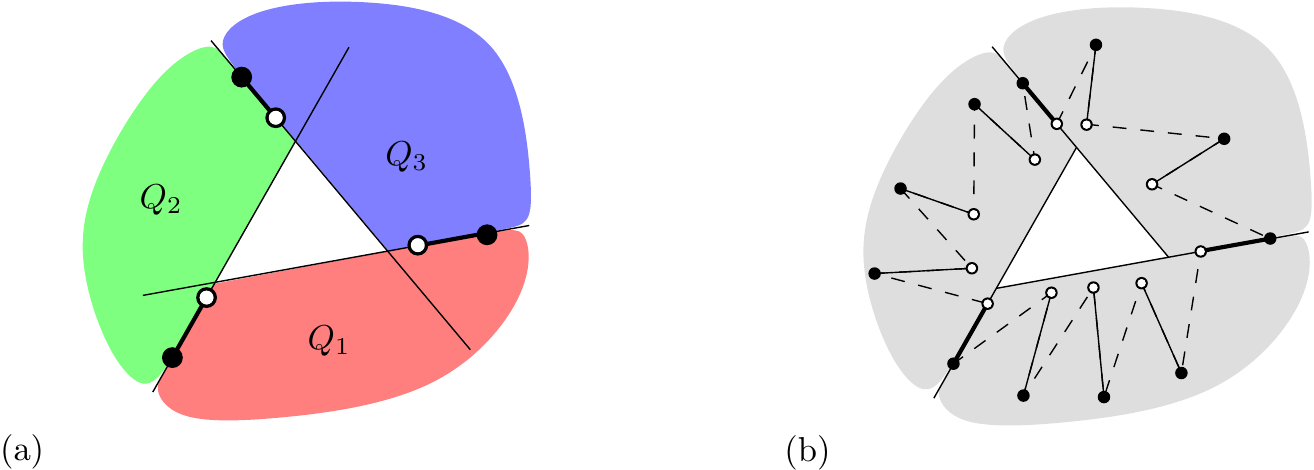}
\caption{Illustration to the proof of Theorem~\ref{thm:typeCNotUnique}.}
\label{fig:regions}
\end{figure}

All segments in a region $Q_i$ together with the two defining segments are of linear type
(indeed, they have no chromatic cut but have both colors on the boundary of the convex hull).
Thus, by Lemma~\ref{fig:open_alt_path}, in each region there is an alternating path from the \x-point of the left bounding segment to the \ci-point of the right bounding segment (or vice versa). 
The union of the three paths forms an alternating polygon and thus we have found a different compatible BR-matching $M'$. 
If we choose the paths in the other direction (\ci-point of the left bounding segment to the \x-point of the right bounding segment), 
we get another BR-matching $M''$.
\end{proof}

%%%%%%%%%%%%%%%%%%%%%%%%%%%%%%%%%%%%%%%%%%%%%%%%%%%%%%%%%%%%%%%%%%%%%%%%%%%%%%%%%%%%%%%%%%%%%%%%%

Now we study in more detail the relation \ord\ for circular matchings.
In the proof of Theorem~\ref{thm:typeCNotUnique} we saw that a circular matching
is a union of three linear matchings (see Fig.~\ref{fig:regions}~(b)).
In the next Proposition we prove that in fact it is a union of 
\textit{two} linear matchings.

\begin{lemma}\label{lem:foursets}
Let $M$ be a circular matching, and let $B$ be a segment of $M$.
%Consider
 The matchings
\begin{align*}
M_B^R&=\{X \in M: B \triangleleft X \},\\
M_B^{R+}&=\{X \in M: B \triangleleft X \}\cup\{B\},\\
M_B^L&=\{X \in M: X \triangleleft B \},\\
M_B^{L+}&=\{X \in M: X \triangleleft B \}\cup\{B\}.
\end{align*}
%These four matchings
are not empty, and they are of linear type.
\end{lemma}

\begin{proof}
Consider first the matching $M_B^{R+}$.
Since it contains $B$, it is non-empty.
Since it is a submatching of $M$, it has no chromatic cut.
Both the $\circ$- and the $\bullet$-end of $B$ belong to the boundary of its convex hull;
therefore $M_B^{R+}$ must be of linear type .
Similarly, $M_B^{L+}$ is of linear type.

If $M_B^{R}$ is empty, then $M_B^{L+} = M$,
which is impossible since $M$ is of circular type, and $M_B^{L+}$ of linear type.
Now, since $M_B^{R+}$ is of linear type, and $M_B^{R}$ is a subset of this matching,
$M_B^{R}$ is of linear type as well
(see the remark after Theorem~\ref{thm:main}).
The proof for $M_B^{L}$ is similar.
\end{proof}

\begin{corollary} \label{no-min-max}
  The relation \ord\ in a matching of circular type has no minimum or maximum element:
  \begin{align*}
\forall B\colon \exists A\colon A \ord B\\
\forall B\colon \exists A\colon B \ord A
  \end{align*}
\end{corollary}
\begin{proof}
  For such an element $B$,
$M_B^{L}$
or $M_B^R$
would be empty.
\end{proof}

%\begin{comment}
\begin{lemma}\label{lem:circ_triple}
Let $M$ be a circular matching.
Let $B$ be any segment of $M$.
Let $A$ be the minimum (with respect to $\triangleleft$) element of $M_B^{L}$,
and let $Z$ be the maximum element of $M_B^{R}$.
Then the triple $\{A, B, Z\}$ is a circular matching (a 3-star).
\end{lemma}

\begin{proof}

If $M$ is of size $3$, that is, $M = \{A, B, Z\}$, there is nothing to prove.
So, we assume that there is at least one more segment in $M$.
Assume without loss of generality that $M_B^{R}$ contains at least one segment in addition to $Z$.

Let $D$ be a segment of $M$ such that $D \triangleleft A$
(such a segment exists by Proposition~\ref{lem:foursets}).
Since $A$ is the minimum element of $M_B^{L}$,
we have $D \in M_B^{R}$, that is, $B \triangleleft D$.

If $D=Z$ then we have $Z \triangleleft  A \triangleleft B \triangleleft Z $:
that is, the relation $\triangleleft$ in the triple $\{A, B, Z\}$ is not linear;
therefore $\{A, B, Z\}$ is of circular type.

Suppose now that $D \neq Z$, and
consider the matching $\{A, B, D, Z\}$.
We have $D \triangleleft A \triangleleft B  \triangleleft D $.
So, the relation $\triangleleft$ in the matching $\{A, B,  D, Z\}$ is not linear;
therefore, $\{A, B, D, Z\}$ is of circular type.
Now, by Lemma~\ref{lem:foursets}, some segment in $\{A, B, D, Z\}$ must lie to the right of $Z$.
Since $B \triangleleft Z$ and $D \triangleleft Z$, we have $Z \triangleleft A$.
So, we have $Z \triangleleft A \triangleleft  B \triangleleft Z $,
and this means that $\{A, B, Z\}$ is of circular  type.
\end{proof}

%\end{comment}

 We shall show that if $M$ is a circular matching, then
 there exists a natural \emph{circular order} of its members.
% We recall mind that
A circular (or cyclic) order is a ternary relation
% $[*,*,*]$ 
which models the ``clockwise'' relation among elements arranged on a cycle. 
A standard way of constructing a circular order from 
$j$ linear orders
$A_{11} \leq A_{12} \leq \dots \leq A_{1i_1} $,
$A_{21} \leq A_{22} \leq \dots \leq A_{2i_2} $,
$\dots,$
$A_{j1} \leq A_{j2} \leq \dots \leq A_{ji_j} $
is their ``gluing'':
we say that $[X, Y, Z]$ (and, equivalently, $[Y, Z, X]$ and $[Z, X, Y]$)
if we have $X \leq \dots \leq Y \leq \dots \leq Z$
in the sequence
%\begin{equation}\label{eq:cyclic_def}
\[ A_{11} \leq A_{12} \leq \dots \leq A_{1i_1} \leq A_{21} \leq A_{22} \leq \dots \leq A_{2i_2} \leq \dots
 \leq A_{j1} \leq A_{j2} \leq \dots \leq A_{ji_j} \leq A_{11} \]
% \end{equation}
 In this line $\leq$ relates only to pairs of neighbors; in particular, it is not transitive in this line.

We fix $B \in M$ and apply this procedure on $M_B^{L+}$ and $M_B^{R}$ in which $\triangleleft$ is linear by Lemma~\ref{lem:foursets}.
%Specifically, let $M$ be a circular matching, and let $B \in M$.
Let $A_1, A_2, \dots, A_k$ be the segments of $M_B^{L}$ labeled so that $A_1 \triangleleft A_2 \triangleleft \dots \triangleleft A_k $,
and let $C_1, C_2, \dots, C_m$ be the segments of $M_B^{R}$ labeled so that $C_1 \triangleleft C_2 \triangleleft \dots \triangleleft C_m $.
By Lemma~\ref{lem:circ_triple} we have $C_m \triangleleft A_1$.
Thus, we consider the circular order $[*,*,*]$ induced by
\begin{equation}\label{eq:cyclic}
B \triangleleft C_1 \triangleleft C_2 \triangleleft \dots \triangleleft C_m \triangleleft A_1 \triangleleft A_2 \triangleleft \dots \triangleleft A_k \triangleleft B.
\end{equation}
That is, for $X, Y, Z \in M$ we have $[X, Y, Z]$ 
(and, equivalently $[Y, Z, X]$ and $[Z, X, Y]$)
if and only of we have in \eqref{eq:cyclic}
$X \triangleleft \dots \triangleleft Y \triangleleft \dots \triangleleft Z$, or
$Y \triangleleft \dots \triangleleft Z \triangleleft \dots \triangleleft X$, or
$Z \triangleleft \dots \triangleleft X \triangleleft \dots \triangleleft Y$.
Notice that %$[X, Y, Z]$, $[Y, Z, X]$ and $[Z, X, Y]$ are equivalent, and 
we always have either $[X, Y, Z]$ or $[X, Z, Y]$ (but never both).

The circular order $[*,*,*]$ will be referred to as the \textit{canonical circular order on $M$}.
The next results describe the geometric intuition beyond this definition:
we shall see that $[X,Y,Z]$ means in fact that these segments appear in this order clockwise.
Moreover, we shall see that the definition of $[*,*,*]$ does not depend on the choice of $B$.

\begin{lemma}\label{lem:circ_triples}
Let $M$ be a circular matching,
and let $X, Y, Z \in M$.
Then we have $[X, Y, Z]$ if and only if 
  at least two among the following three conditions hold:
$X \triangleleft Y $;
$Y \triangleleft Z $;
$Z \triangleleft X $.
\end{lemma}

If all three conditions hold, then $\{X,Y,Z\}$ is a $3$-star;
and if exactly two among the statement hold, then $\{X,Y,Z\}$ is a linear matching.
All possible situations for $[X, Y, Z]$ (with respect to $\triangleleft$)
appear in Fig.~\ref{fig:cyclic}.

\begin{figure}[h]
$$\includegraphics{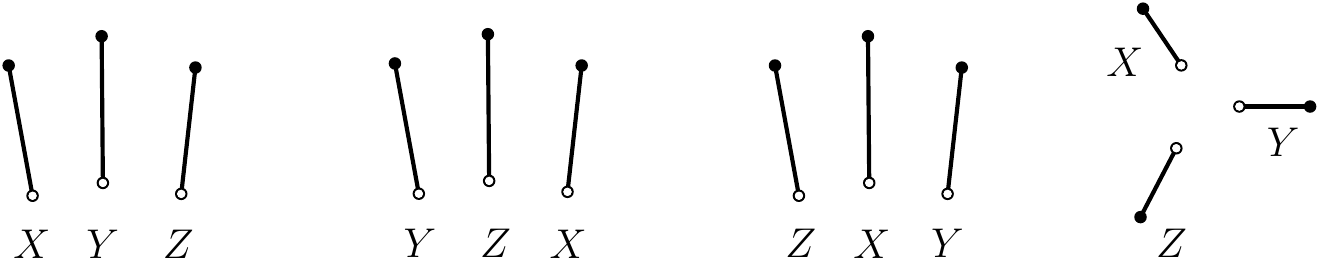}$$
\caption{Possible configurations of three segments that satisfy $[X, Y, Z]$.}
\label{fig:cyclic}
\end{figure}

\begin{proof}
The segment $B$ from the definition of $[*, *, *]$
is the maximum element of $M_B^{L+}$. Therefore, it is 
convenient to denote $A_{k+1}=B$. Now we have four cases.

\begin{itemize}
\item Case $1$: $X, Y, Z \in M_B^{L+}$.

In this case $\{X,Y,Z\}$ is of linear type (by Lemma~\ref{lem:foursets}).
Therefore either one or two of the conditions hold.
If exactly two conditions hold: assume without loss of generality that $X \triangleleft Y \triangleleft Z $.
Since $A_1 \triangleleft \dots \triangleleft A_{k+1}$ is a linear order in $M_B^{L+}$,
we have $X=A_\alpha, Y=A_\beta, Z=A_\gamma$ for some $1 \leq \alpha < \beta <\gamma \leq k+1$.
Now we have $[X,Y,Z]$ by definition.
If exactly one condition holds: assume that it is $X \triangleleft Y $;
then we have $Y \triangleleft Z \triangleleft X $, which implies ``not $[X, Y, Z]$''.

\item Case $2$: two members of $\{X,Y,Z\}$ belong to $M_B^{L+}$, and one to $M_B^{R}$.
Assume without loss of generality that $X, Y \in M_B^{L+}, Z \in M_B^{R}$
and that $X \triangleleft Y$.

Then we have  $X=A_\alpha, Y=A_\beta$ for some $\alpha < \beta$
and $Z = D_\gamma$ for some $\gamma$, and so $[X, Y, Z]$.

At the same time in this case at least two of the conditions hold:
indeed, assume $X \triangleleft Z\triangleleft Y$. 
Then $B$ is distinct from $X, Y, Z$ (in particular, $B \neq Y$ because $B \triangleleft Z$).
Now, in the matching $\{X, Y, Z, B\}$ there is a minimum element, $X$, 
but there is no maximum element. Therefore, $\{X, Y, Z, B\}$ is neither of linear nor of circular type --- a contradiction.
\item Case $3$: one member of $\{X,Y,Z\}$ belongs to $M_B^{L+}$, and two to $M_B^{R}$,
and Case $4$: all the members of $\{X,Y,Z\}$ belong to $M_B^{R}$,
are similar to cases $2$ and $1$. Therefore, we omit their proofs. 
\end{itemize}
\end{proof}

\begin{corollary}\label{cor:circ_indep}
The canonical circular order
%circular ordering defined by~\eqref{eq:cyclic} 
doesn't depend on the choice of $B$.
\end{corollary}

\begin{proof}
Indeed, if another choice of $B$ gave another circular order,
there would be a triple that belongs to one of them and doesn't belong to another.
However, in Lemma~\ref{lem:circ_triples} we saw an equivalent definition that only depends on relations between triples of segments.
\end{proof}

\begin{lemma}\label{lem:circ_succ}
Let $M$ be a circular matching,
and let $X \in M$.
Then the immediate successor of $X$ %in~\eqref{eq:cyclic}
the canonical circular order
is the minimum element of $M_X^{R+}$.
\end{lemma}

\begin{proof}
This is immediate for $B$ (as in definition of $[*,*,*]$), and, since we saw in Corollary~\ref{cor:circ_indep}
that the circular order $[*, *, *]$ does not depend on the choice of $B$, this is true for all segments.
\end{proof}

Lemmas~\ref{lem:circ_triples} and~\ref{lem:circ_succ} show that the canonical circular order 
describe the combinatorial structure of circular matchings in a natural way, 
similarly to that in which $\triangleleft$
describes the structure of linear matchings.
In Subsection~\ref{sec:triangle}
we'll provide a finer classification of relations $\triangleleft$
realizable in circular matchings.

\section{Summary of the Proof of the Characterization Theorem~\ref{thm:main} about Unique BR-Matchings and Theorem~\ref{thm:circularMatchings} about Circular Matchings}
\label{sec:summary}
%%%%%%%%%%%%%%%%%%%%%%%%%%%%%%%%%%%%%%%%%%%%%%%%%%%%%%%%%%%%%%%%%

%\texttt{AAA: Summary for both proofs (Theorems 2 and 3) here?
%Rename the section to "Summary: the Proofs of main results", or smth. similar?}
%
%\bigskip

We summarize the proofs of both Theorems.

We start with the equivalence of all five conditions in Theorem~\ref{thm:main}.
Equivalence of conditions~$2, 3, 4, 5$ is proven in Lemma~\ref{lem:typeL}.
%Next, we have $2 \Rightarrow 3$ by Lemma~\ref{lem:chCCut}.
Finally, $2 \Rightarrow 1$ (if a BR-matching $M$ is of linear type, then it is unique) is proven in Theorem~\ref{thm:type2unique};
and $1 \Leftrightarrow 2$ (if a BR-matching $M$ is unique, then it is of linear type) follows from Corollary~\ref{cor:bal} (if $M$ is unique, then it has no chromatic cut),
Lemma~\ref{lem:otherPS_v2} (if $M$ has no chromatic cut, then it is either of linear or circular type),
and Theorem~\ref{thm:typeCNotUnique} (if $M$ is of circular type, then it is not unique).

\smallskip

Consider the conditions~$1, 2, 3$ and Properties p1 and p2 in 
Theorem~\ref{thm:circularMatchings}.
All three conditions imply that $M$ has no chromatic cut and exclude that
$M$ is of linear type. Thus in all cases $M$ must be a circular matching.
Property p1 is explained in Section~\ref{sec:type_c} and 
p2 is proved by Theorem~\ref{thm:typeCNotUnique}.
%"$1\Rightarrow 2$": 
%Clearly, if $M$ is a circular matching the sidednessrelation \ord\ is not ordered linearly but is total by Lemma~\ref{lemma-total-order}.
%
%"$2\Rightarrow 3$": 
%As \ord\ is total the pattern shown in Figure~\ref{fig:allConfig}~(a) and~(b) cannot occur, because these segments are not comparable w.r.t. the sidedness relation \ord. 
%If the pattern of Figure~\ref{fig:allConfig}~(c) would not appear \ord\ would be a linear order -- a contradiction.
%
%"$3\Rightarrow 1$": 

\section{Miscellaneous}\label{sec:misc}

A \emph{parallel matching} is a BR-matching that consists of parallel segments.
As we saw in Theorem~\ref{thm:main},
quasi-parallel matchings generalize parallel matchings
%(that is, BR-matchings that consist of parallel segments)
in the sense that they are exactly the BR-matchings for which the relation $\triangleleft$ is a linear order.
Similarly, circular matchings generalize \emph{radial matchings} --
BR-matchings whose members lie on distinct rays with a common endpoint $O$
and oriented away from $O$.

In this section we study how far quasi-parallel (resp., circular) matchings
generalize parallel (resp., radial) matchings, in two aspects.
In Subsection~\ref{sec:parallel} we consider order types,
and in Subsection~\ref{sec:triangle} we deal with $\triangleleft$ relations realizable in such matchings.

\subsection{Order types in parallel vs.\ quasi-parallel matchings}\label{sec:parallel}

%As we showed in Theorem~\ref{thm:main},
%quasi-parallel matchings generalize
%\emph{parallel matchings}
%(that is, BR-matchings that consist of parallel segments)
%in the sense that they are exactly the BR-matchings for which the relation $\triangleleft$ is a linear order.
Since, as mentioned above, 
quasi-parallel matchings generalize parallel matchings,
it is natural to ask whether
all order types (determined by orientations of triples of points)
of bichromatic point sets with a unique BR-matching
are realizable by corresponding endpoints of a parallel matching.

We construct an example that shows
that the answer to this question is negative.
The construction is based on the following observation.

\begin{observation}\label{obs:long}
Let $A, B, C$ be three parallel vertical segments such that  $A \triangleleft B \triangleleft C$.
Denote by $a_1, b_1, c_1$ the upper ends, and by $a_2, b_2, c_2$ the lower ends of the corresponding segments.
If the triple $[a_1, b_1, c_2]$
is oriented counterclockwise,
and the triple $[a_2, b_2, c_1]$
clockwise,
then
$B$ is shorter than $A$.
\end{observation}

\begin{proof}
The conditions mean that $c_2$ is situated above the line $a_1b_1$, and
$c_1$ below the line $a_2b_2$.
However, if $B$ is not shorter than $A$, then the
wedge that should contain $C$ is situated to the left of $A$.
Thus, $A \triangleleft C$ is impossible.
See Fig.~\ref{fig:long} for illustration.
\end{proof}

\begin{figure}
%$$\includegraphics[width=220pt]{figures/long}$$
$$\includegraphics{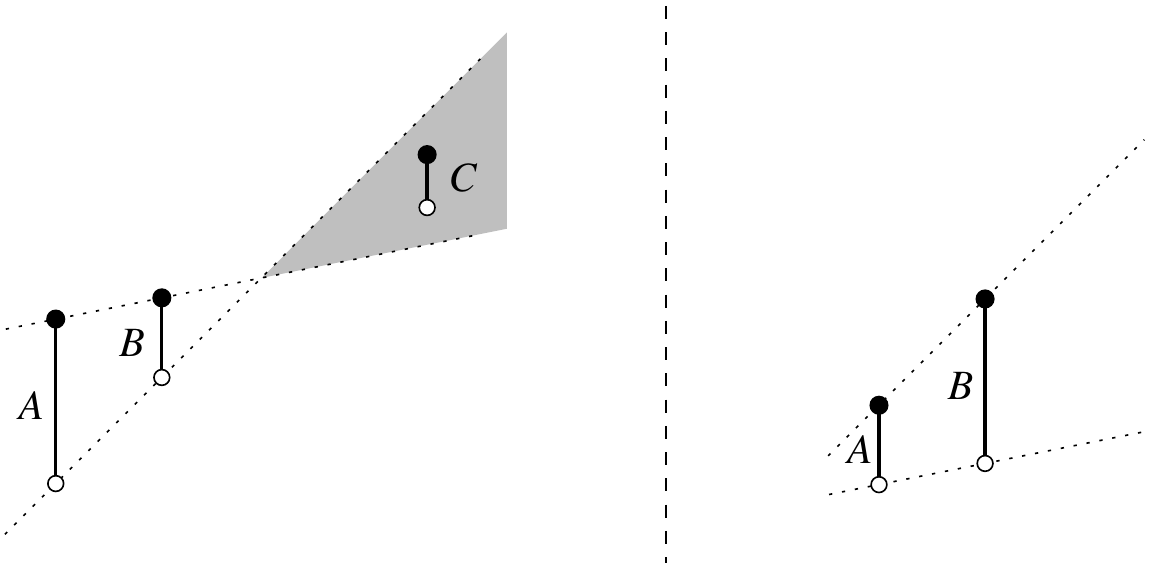}$$
\caption{Illustration to Observation~\ref{obs:long}.}
\label{fig:long}
\end{figure}

Now, the construction goes as follows.
%Refer to Fig.~\ref{fig:example}.
Consider three pairs of parallel (auxiliary) lines with slopes, say, $0^\circ$, $60^\circ$, and $120^\circ$,
and three vertical segments $A_0, B_0, C_0$, as shown in
 Fig.~\ref{fig:example}a. % the left side of the figure.
Change slightly the slopes of the lines so that
each pair will intersect as indicated schematically in the right part,
and so that the new segments $A, B, C$ whose endpoints are intersection points of the modified lines
are almost vertical.
Add vertical segments in the wedges formed by the auxiliary lines, as
shown in
Fig.~\ref{fig:example}b. % the right part.
This can be done so that the new matching (consisting of six segments) is quasi-parallel;
denote it by $M$.
\begin{figure}[htb]
  \centering
\includegraphics%[width=350pt]
{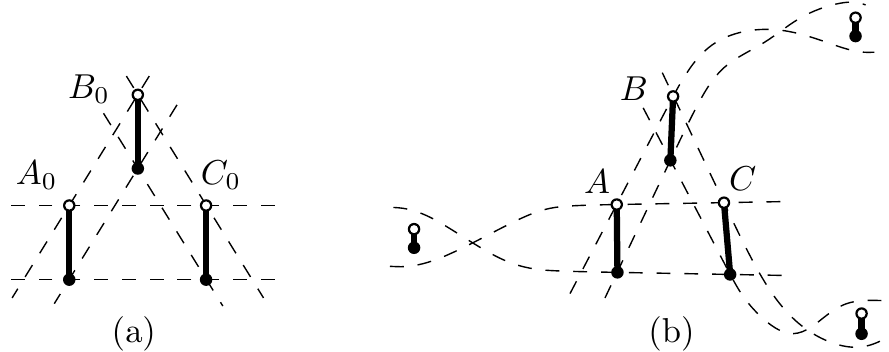}
\caption{The construction of a ``non-parallelizable'' quasi-parallel matching.}
\label{fig:example}
\end{figure}

Now, assume that there exists a parallel matching $M'$ with endpoints of the same order type,
and denote by $A', B', C'$ the segments that correspond in $M'$ to $A, B, C$.
Then, according to Observation~\ref{obs:long},
$A'$ is longer than $B'$,
$B'$ is longer than $C'$, and
$C'$ is longer than $A'$.
This is a contradiction. \hfill\qedsymbol

In the same manner radial matchings, as defined above, do not capture the order type.
To see this let $M$ be a non-parallelizable linear BR-matching and assume there exists a
radial representation $M'$ with the same order type and apex $O$. Project the point $O$ to infinity with a projective mapping such that the matching becomes parallel. As projective mappings conserve the order type we parallelized $M$ -- a contradiction. \hfill \qedsymbol

%\texttt{AAA: (1) Here we have QED without having Lemma/Proposition/etc.
%But I am not sure that adding an environment is good since then we have to formulate strictly the statement, and it will be very long and cumbersome.
%Then maybe we'll just remove QED? (2) A SHORT REMARK that the same could be done for circular matchings?}

\subsection{Sidedness relation in circular matchings}\label{sec:triangle}

If $M = \{A_0, A_1, \dots, A_{n-1}\}$ is a linear matching,\footnote{
We denote the segments by $A_0, \dots, A_{n-1}$
rather than by $A_1, \dots, A_{n}$
because modular arithmetic will be used in this section.
}
and we know that $A_0 \triangleleft A_1 \triangleleft \dots \triangleleft A_{n-1}$,
then the relation $\triangleleft$ is completely determined (since it is linear by Lemma~\ref{lem:typeL}).
In contrast, for matchings of $n$ segments of circular type,
there are several relations $\triangleleft$ that satisfy $A_0 \triangleleft A_1 \triangleleft A_2 \triangleleft \dots \triangleleft A_{n-1} \triangleleft A_0$.
In Theorem~\ref{thm:enum_c} we enumerate such relations;
its proof also provides us with a classification of circular matchings in the sense of relations \ord\ realizable in such matchings.

\begin{lemma}\label{lem:antipodes}
Let $M$ be a circular matching,
and assume that its canonical circular order is induced by $A_0 \triangleleft A_1 \triangleleft A_2 \triangleleft \dots \triangleleft A_{n-1} \triangleleft A_0$.
Fix $i \in \{0, 1, \dots, n-1\}$.
Then there exist a unique $j\in \{0, 1, \dots, n-1\}$
such that $A_j$ and $A_{j+1}$ (mod $n$) are separated by $g(A_i)$.
\end{lemma}

\begin{proof}
Let $A_j$ be the maximum member of $M_{A_i}^{R}$.
%It follows from Lemma~\ref{lem:circ_succ} that 
Then we have $A_{j+1} \in M_{A_i}^{L}$,
and, thus, $A_j$ and $A_{j+1}$ are separated by $g(A_i)$.
For any other pair of neighboring segments,
either both belong to $M_{A_i}^{L+}$ or to $M_{A_i}^{R+}$,
and, therefore, are not separated by $g(A_i)$.
\end{proof}

The pair of segments $(A_j, A_{j+1})$ as in Lemma~\ref{lem:antipodes} will be called
the \textit{antipodal pair} of $A_i$.
Notice that by Lemma~\ref{lem:circ_triple} such $A_i$, $A_j$ and $A_{j+1}$ form a $3$-star.
We say that $A, B \in M$ are \textit{twins} if they have the same antipodal pair.
Clearly, being twins is an equivalence relation;
the equivalence classes are the maximal sets of twins (we shall call them \textit{T-sets}).
In Fig.~\ref{fig:radial} we have five T-sets: $\{9,0,1\}$, $\{2\}$, $\{3,4,5\}$, $\{6,7\}$, and $\{8\}$.
It is easy to see that any T-set is a linear matching.

\begin{figure}[htb]
  \centering
\includegraphics[]{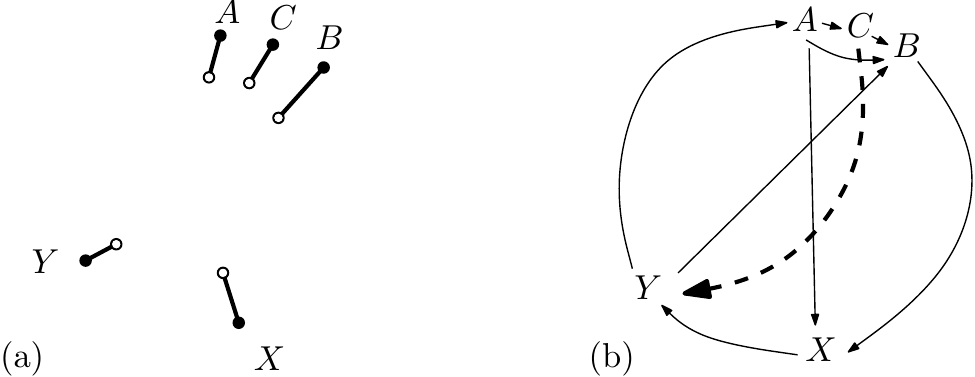}
  \caption{(a) Geometric intuition for Lemma~\ref{lem:twins}. (b) The arrow indicates the sidedness relation \ord. The dashed arrow is impossible.}
  \label{fig:twin}
\end{figure}

\begin{lemma}\label{lem:twins}
Let $M$ be a circular matching,
and let $A, B \in M$ (assumed $A \triangleleft B $) be twins.
Then any $C \in M$ such that $A \triangleleft C \triangleleft B $
is also a twin of $A$ and $B$.

\end{lemma}

\begin{proof}
Let $(X, Y)$ be the antipodal pair of $A$ and of $B$.
Then we have $C \triangleleft  B \triangleleft X \triangleleft Y
\triangleleft A \triangleleft C$,
% \footnote{
%Recall again that $\triangleleft$ is not assumed to be transitive.
%},
$A \triangleleft  X$,  $Y \triangleleft  B$ and $A\triangleleft  B$
(see Fig.~\ref{fig:twin}.).
If $C \triangleleft Y$, then the matching $\{A, C , B, Y\}$ has a maximum ($B$), but has no minimum --- a contradiction.
Thus, $Y \triangleleft C$, and, similarly, $C \triangleleft X$.
Therefore, $g(C)$ separates $X$ and $Y$.
Now it follows from the uniqueness in Lemma~\ref{lem:antipodes}
that $(X, Y)$ is the antipodal pair for $C$.
\end{proof}

We say that a circular matching is \emph{basic} if it has no twins,
%(except the trivial case of being a segment a twin of itself),
or equivalently, if all T-sets consist of one segment.
We first classify $\triangleleft$ relations of basic matchings.

\begin{lemma}\label{lem:basic}
Let $M$ be a \textbf{basic} circular matching,
with the circular order induced by $A_0 \triangleleft A_1 \triangleleft A_2 \triangleleft \dots \triangleleft A_{n-1} \triangleleft A_0$.
Then, for each $A_i \in M$, we have $\left|M_{A_i}^L\right|=\left|M_{A_i}^R \right|$.
\end{lemma}

\begin{proof}
Let $B \in M_{A_i}^R$, and let $(X, Y)$ be the antipodal pair of $B$.
We claim that $X, Y \in M_{A_i}^{L+}$.
Indeed, if $X, Y   \in M_{A_i}^{R}$, then 
$\{A_i, X, Y\}$ is a linear matching, contradicting what was observed
after the definition of the antipodal pair;
and if $X\in M_{A_i}^{R}, Y\in M_{A_i}^{L}$ then $(X, Y)$ is the antipodal pair of $A_i$,
which is impossible since in such a case $A_i$ and $B$ are twins.

Assume for contradiction and without loss of generality that $\left|M_{A_i}^L\right|<\left|M_{A_i}^R\right|$.
Then we have more segments in $M_{A_i}^R$
%than pairs of neighboring segments that can be their antipodal pairs.
than their potential pairs of neighboring segments.
Therefore, there exist distinct segments $B, C \in M_{A_i}^R$
that have the same antipodal pair,
%the lines $g(B)$ and $g(C)$ pass between the same pair of the members of $M_{A_i}^{L+}$.
%Moreover, all segments of $M_{A_i}^R$ between $B$ and $C$ satisfy this
%property.
and thus, they are twins -- a contradiction.

It follows that $n$ is necessarily {odd}, and that
$M_{A_i}^R= \left\{ A_{i+1}, A_{i+2}, \dots, A_{i+\frac{n-1}{2}}  \right\}$ and
$M_{A_i}^L=\left\{A_{i-\frac{n-1}{2}}, \dots, A_{i-2}, A_{i-1}\right\}$
(mod $n$).
In particular, this means that for basic circular matchings, 
the relation $\triangleleft$ is determined uniquely by
$A_0 \triangleleft A_1 \triangleleft A_2 \triangleleft \dots \triangleleft A_{n-1} \triangleleft A_0$.
\end{proof}

Fig.~\ref{fig:basic} shows basic matchings of sizes 3, 5, 7.
\begin{figure}[htb]
  \centering
\includegraphics[width=140mm]{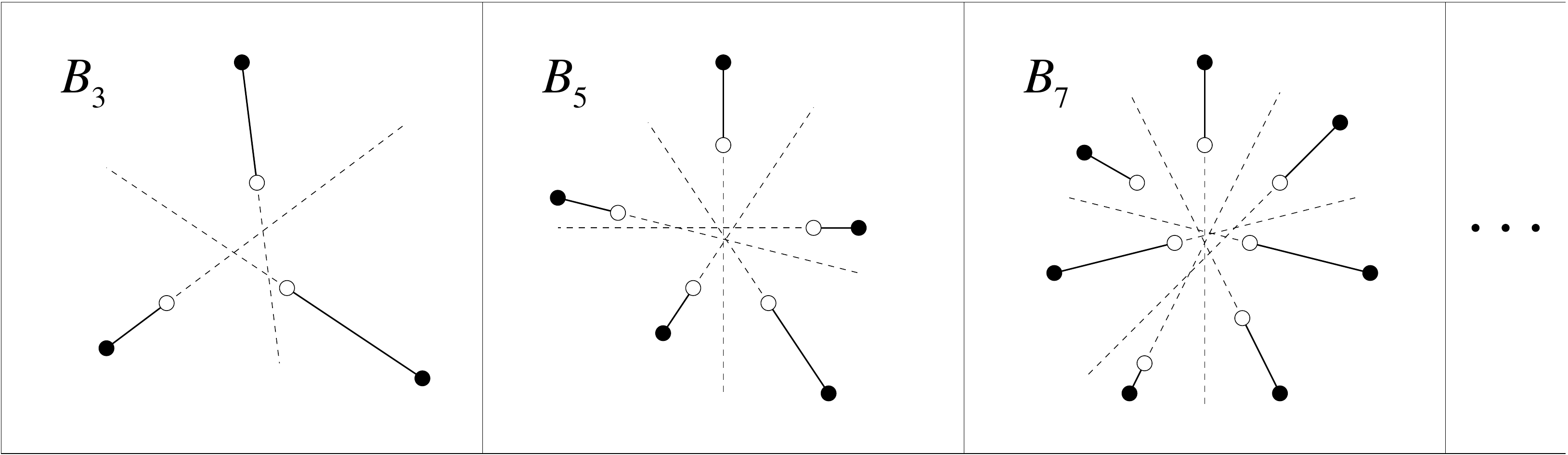}
  \caption{Basic matchings of sizes 3, 5, and 7.}
  \label{fig:basic}
\end{figure}

\begin{theorem}\label{thm:enum_c}
The number of sidedness relations $\triangleleft$
realizable by circular matchings $\{A_0, A_1,\allowbreak \dots,\allowbreak A_{n-1}\}$
that satisfy $A_0 \triangleleft A_1 \triangleleft A_2 \triangleleft \dots \triangleleft A_{n-1} \triangleleft A_0$
is $2^{n-1}-n$.
\end{theorem}

\begin{proof}
Consider a circular matching $M$ of size $n$.
Assume that there are $k$ T-sets.
Choose one segment from each T-set as follows:
in the T-set which contains $A_0$, we choose $A_0$;
in all other T-sets we choose the minimum element (with respect to $\triangleleft$).
The chosen segments form a basic matching $B_k$.
Recall that, by Lemma~\ref{lem:basic}, the relation $\triangleleft$ between the members of $B_k$ is uniquely determined.
Now $M$ can be obtained from $B_k$
by recovering the twins of the chosen segments.
The relation $\triangleleft$ for $M$
is then only determined by sizes of T-sets except that of $A_0$;
and for the T-set of $A_0$ it also matters how many segments lie to the left of $A_0$ and how many to the right.
Thus, we have $k+1$ ``regions'' for adding twins,
and this is the combinatorial problem of 
choosing a multiset of size $n-k$ from $k+1$ elements.
%If $M$ is not basic,
%then it is possible to delete segments that have twins
%until one obtains a basic matching.
%Therefore any circular matching of size $n$
%can be obtained from a basic matching of size $k \leq n$
%by producing a number of twins for some of its segments.
%To this matter, we fix one of the segments of a basic matching $B_k$ to be $A_0$.
For fixed $k$, the corresponding generating function
%(which, thus, enumerates the contribution of the uniquely defined $\triangleleft$ on basic matchings of size $k$
%in construction of all realizable relations $\triangleleft$) 
is
\[x^k(1+x+x^2+x^3+\dots)^{k+1} = x^k \left(\frac{1}{1-x}\right)^{k+1},\]
%and the contribution of $\triangleleft$ on all basic matchings is
%given by the sum
the summation over all odd $k\ge 3$ gives
\begin{equation}\label{eq:gf}
 \frac{x^3}{(1-x)^{4}} +  \frac{x^5}{(1-x)^{6}} +  \frac{x^7}{(1-x)^{8}} + \dots =
% \]
%\[
% = \frac{x^3}{(1-x)^{4}} \cdot \frac{1}{1-\frac{x^2}{(1-x)^2}} =
%  \frac{x^3}{(1-x)^{4}} \cdot \frac{(1-x)^2}{1-2x} =
% x \cdot \frac{x^2}{(1-x)^{2}(1-2x)}=
% \]
%\[ =
\frac{x}{1-2x} - \frac{x}{(1-x)^2}.
%\]
\end{equation}
Since
\[\frac{1}{1-2x}
% = 1 + 2x + 4x^2 + 8x^3 + \dots
= \sum_{n\geq 0} 2^n x^n
%\]
%and
%\[
\ \ \ \ \ \mathrm{and} \ \ \ \ \
\frac{1}{(1-x)^2} = %1 + 2x + 3x^2 + 4x^3 + \dots =
\sum_{n\geq 0} (n+1) x^n,\]
the coefficient of $x^n$ in the generating function~\eqref{eq:gf} is $2^{n-1} - n$.
\end{proof}

\begin{figure}[h]
$$\includegraphics[width=120mm]{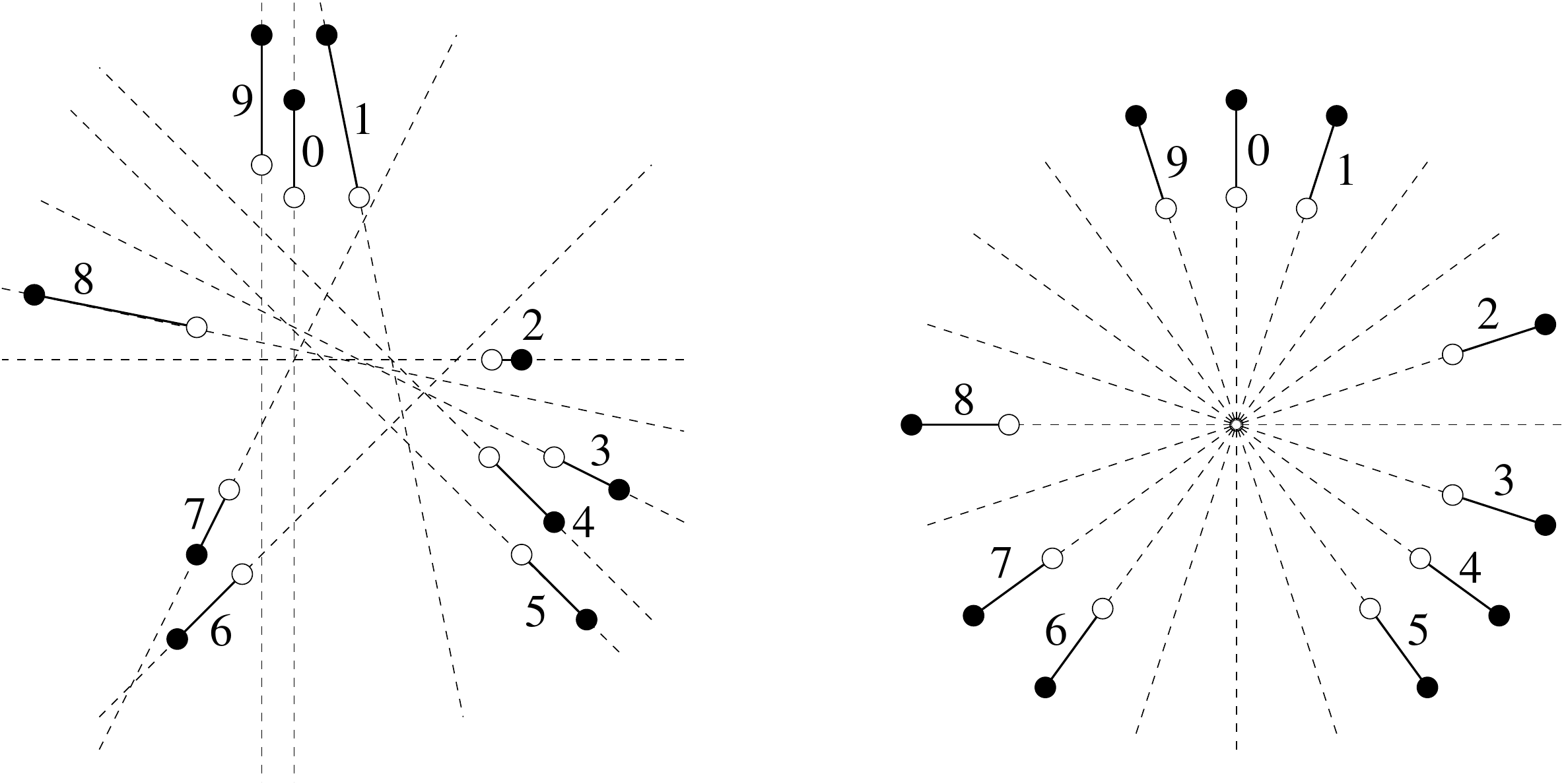}$$
\caption{A circular matching and its ``standardization''.}
\label{fig:radial}
\end{figure}

It is easy to see from the reasoning above that any $\triangleleft$ relation which is realizable
with circular matchings can be realized by a radial matching;
%in which the supporting lines of the segments intersect in a common point $O$
furthermore, it is possible to take the lines such that the angle between each pair of adjacent lines will be $\pi / n$,
and the endpoints the segments lie on two fixed circles with center $O$,
%(we'll call such a matching \emph{a standard representation of $\triangleleft$};
see Fig.~\ref{fig:radial}).
Now the formula $2^{n-1}-n$ becomes especially clear:
consider $n$ lines passing through a common point $O$.
For each line, there are two choices on which ray we put a segment
(except the fixed $A_0$). Thus we have $2^{n-1}$ matchings:
$n$ of them are of linear type, and the others are of circular type.

\subsection{Description in terms of point sets}\label{sec:points_vs_matchings}

We described point sets with unique matchings in terms of a given
matching $M$ rather than in terms of the set $F$ itself.
It would be nice to characterize the points sets $F$ directly, for
example by forbidden patterns \textit{of points}.
However, such a characterization is impossible. 
% to describe the point sets with unique (respectively, non-unique) matching by

Suppose that there are is a collection of patterns of points (of two colors)
such that $F$ has unique matching if and only if $F$ avoids these patterns.
Equivalently, $F$ has several matchings if and only if $F$ contains any of these patterns.
However, in such a case we can duplicate all the members of $F$:
for each $p_i\in F$ we add a point $p_i'$ so that 
$p_i$ and $p_i'$ are of opposite colors,
all the segments $p_ip_i'$ are parallel (including orientation),
and the new set is in general position
(see Fig.~\ref{fig:double}).
Then the matching that consists of all the segments $p_ip_i'$
is a (quasi-)parallel linear matching, and thus is 
a unique matching of the new set,
while it contains the assumed pattern(s).

We can actually move the additional points as far away as we
like. Thus, even a more ``local'' characterization, that a certain
convex region should contain some pattern and no other points, is impossible.

\begin{figure}[hpt]
\includegraphics{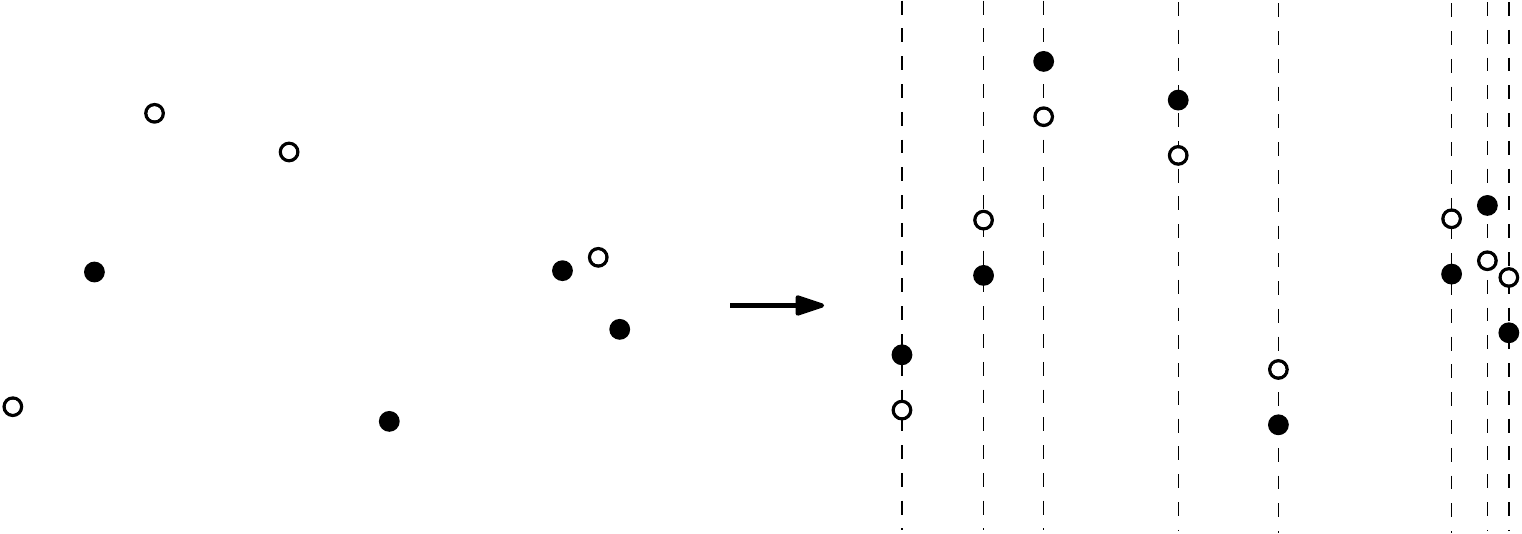}
\caption{Illustration of the duplication described in Subsection~\ref{sec:points_vs_matchings}.}
\label{fig:double}
\end{figure}

On the other hand,
suppose that there are is a collection of  point patterns
such that $F$ has several matchings if and only if $F$ avoids these patterns.
Equivalently, $F$ has unique matching if and only if $F$ contains any of these patterns.
However, in such a case we can take this matching and add one more segment to obtain a BR-matching with a chromatic cut. 
So, the new point set will have more than one matching while it
contains the assumed pattern(s).
As above, the additional segment can be placed arbitrarily far away.

\section{Algorithms}\label{sec:algo}

In this section we describe several algorithms. The first checks
whether a given point set $F$ has a unique BR-matching.
This algorithm is based on yet another characterization of unique BR-matchings.
The second checks if a given BR-matching is circular.
Applying these algorithms together, we can check if a given matching has a chromatic cut.

\begin{definition}
  A BR-matching $M$ has the \emph{\dd property} with respect
to the segments $A, B\in M$ ($A \neq B$) if
%\begin{enumerate}
%\item
  $A$ and $B$
%  lie on the boundary of $\\CH(S)$, and these
are the only
segments from $M$ on $\partial \CH(F)$.
%and they appear in opposite orientations when walking on $\partial\CH(F)$.
%\end{enumerate}
\end{definition}

%?? find a better name ??

%?? maybe we can also use the STRONG \dd property that
%in addition,
%the vertices on $\\CH(S)$ are colored in two intervals ??

\begin{theorem}\label{x}
Let $M=\{A_1,A_2,\dots,A_{n}\}$ be a BR-matching such that
%  Let $A_1, A_2, \dots, A_n$ be ordered so that
$A_1 \LL A_2 \LL  \cdots \LL A_n$.
Then the following conditions are equivalent:
\begin{enumerate}
\item
 $M$ is the unique BR-matching.

\item
For every $i<j$,
every subset $S\subseteq
\{A_i,A_{i+1},\dots,A_{j}\}$
with $A_i,A_j\in S$ has the \dd property
for $A_i$ and $A_j$.

\item
\label{algo}
For every $j>1$,
the set $\{A_1,A_2,\dots,A_j\}$ has the \dd property for $A_1$ and $A_j$;
and for every $i<n$,
the set $\{A_i,A_{i+1},\dots,A_n\}$
has the \dd property
for $A_i$ and $A_n$ .

\end{enumerate}
\end{theorem}

%\footnote{
Note that the relation $\LL$ is % in general
not necessarily transitive. So the assumption of the theorem does not imply
$A_i\LL A_j$ for $i<j$.

\begin{proof}

  $``1 \Rightarrow 2''$ %(and $``1 \Rightarrow 3''$)
 follows directly from Theorem~\ref{thm:main},
Condition~\ref{type_l}, together with the remark after the theorem
that the condition is implied for all subsets.

  $``2 \Rightarrow 3''$ is clear.

$``3 \Rightarrow 1''$:
% %
%   First of all, the {\dd property} with respect to two segments
%   $A, B\in M$ implies that all segments $S\ne A$ lie on the right side
%   of $g(A)$ and all segments $S\ne B$ lie on the left side of $g(B)$,
%   or vice versa.
%
% The set $M=\{A_1,A_2,\dots,A_n\}$ has the \dd property for $A_1$ and $A_n$.
% Therefore
% all segments $A_2,\dots$ lie on the same side of $g(A_1)$.
% Since $A_1\LL A_2$, we know that the seg
%
% Let us assume without loss of generality that all segments
% $A_2,\dots$ lie on the \emph{right} side of $g(A_1)$, and all
% segments $A_i\ne A_n$ lie on the \emph{left} side of $g(A_n)$.
%
Since $\{A_1,A_2,\dots,A_j\}$ has the \dd property for $A_1$ and
$A_j$, all segments $A_1,\dots,A_{j-1}$ lie on the same side of
$g(A_j)$.  Since $A_{j-1}\LL A_j$ by assumption, we know that the
segment $A_{j-1}$ lies left of $A_j$, and hence we conclude that all
segments $A_i$ lie left of $g(A_j)$, for $i<j$.  Similarly, from the
\dd property for $\{A_i,A_{i+1},\dots,A_n\}$ we conclude that the
segments $A_j$ lie right of $g(A_i)$, for $j>i$. These two conditions
together mean that $A_i\LL A_j$ for $i<j$. Therefore,
Condition~\ref{order} of Theorem~\ref{thm:main} holds, and $M$ is unique.
% \qed
%
%  Suppose that $M$ is not unique. Then, by Theorem~\ref{thm:main} (and Lemma ... ?),
%   either $M$ is a matching of type C, or it has a chromatic cut.
%
%   If $M$ is a matching of type C, then $M$ itself has no drum property.
%
%   Suppose now that $M$ has a chomatic cut $\ell$ that crosses $A_i$ and $A_j$ ($i<j$)
%   so that the \x-ends of these segments are on different sides of $\ell$.
%   If $g(A_i)$ crosses $A_j$, then $A_i$ doesn't lie on $\CH(\{A_i, A_{i+1}, \dots, A_n\})$.
%   Similarly, if $g(A_j)$ crosses $A_i$, then $A_j$ doesn't lie on $\CH(\{A_1, \dots, A_{j-1}, A_j\})$.
%   Finally, if $g(A_i)$ and $g(A_j)$ cross in the other rays of different colors (or are antiparallel),
%   then either $g(A_j)$ separates $A_{i}$ and $A_{j-1}$,
%   or $g(A_i)$ separates $A_{j}$ and $A_{i+1}$.
%   In the first case
%   $A_j$ doesn't lie on $\CH(\{A_1, \dots, A_{j-1}, A_j\})$;
%   in the second case
%   $A_i$ doesn't lie on $\CH(\{A_i, A_{i+1}, \dots, A_n\})$.
\end{proof}

From Property~\ref{algo} of Theorem~\ref{x} we can derive a
linear-time algorithm for testing whether $M$ is unique, once the ordering
$A_1 \LL A_2 \LL  \cdots \LL A_n$ has been computed:
We incrementally compute
 $P_j := \CH(\{A_1,A_2,\dots,A_j\})$ for $j=2,\dots, n$ and check the
\dd property as we go.

There is a straightforward incremental algorithm for computing the convex
hull (see, for example, \cite{kmpsy-cerpg-08}), which is the basis for
more elaborate randomized incremental algorithms that work also in
higher dimensions, see~\cite[Chapter 11]{bkos-cgaa-00}.
It extends a convex hull $C$ by a new point $p$ as follows:
\begin{enumerate}
\item[C1.]
 Check whether $p\in C$. If this is the case, stop.
\item[C2.] If not, find a boundary point $q\in\partial C$ that is visible from
  $p$.
\item[C3.] Walk from $q$ in both directions to find the tangents $pq_1$ and
  $pq_2$ from $p$ to $C$.
\item[C4.] Update the convex hull: remove the part between $q_1$ and
  $q_2$ that has been walked over, and replace it with $q_1pq_2$.
\end{enumerate}
Steps C3 and C4 take only linear time overall, because everything that
is walked over is deleted.
The ``expensive'' steps that are responsible for the superlinear
running time of convex hull algorithms are C1 and C2.
However, in our case, we will see that these steps are trivial.
(We extend the convex
hull by inserting not a single point but two points of $A_{j+1}$ at a time.)

Since the \dd property holds for
$\{A_1,A_2,\dots,A_{j+1}\}$ we \emph{know} that the new
points of $A_{j+1}$ don't lie in $P_j$, and since $A_j$ lies on the
boundary
of $P_j$ but not of $P_{j+1}$, we know that $A_j$ is visible from the
points of $A_{j+1}$.  We can start the search in step C3 from there.
This visibility assumption can be checked (as part of checking the \dd
property) in constant time. The overall running time is linear.

In a second symmetric step, we start from the end and compute
$\CH(\{A_i,A_{i+1},\dots,A_n\})$ for $i=n-1,\dots,1$.

\begin{theorem}\label{theorem-algo}
  It can be checked in $O(n\log n)$ time whether a bichromatic
  $(n+n)$-set
has a unique non-crossing BR-matching. % $M$.
\end{theorem}

\begin{proof}
  First we have to compute some BR-matching % $M$.
$M=\{A_1,A_2,\dots,A_{n}\}$.
It is well-known that this can be done by recursive ham-sandwich cuts in $O(n\log n)$ time.
A ham-sandwich cut is a line $\ell$ that partitions a bicolored ($n+n$)-set
such that each open half-plane contains at most $\left\lfloor \frac
{n}{2} \right\rfloor$ points of each color.
If $n$ is odd, $\ell$ must go through a red and a blue point.
We can match these points to each other, and recursively find
a BR-matching in the
$\left(\frac{n-1}2+\frac{n-1}2\right)$-sets in each half-plane.
If $n$ is even, $\ell$ may go through one or two points, but by
shifting $\ell$ slightly we can push these points to the correct side
such that each half-plane
contains an $\left(\frac{n}2+\frac{n}2\right)$-set. We recurse as above.
A ham-sandwich cut
 can be found in linear time~\cite{lms-ahsc-94}. Hence this procedure
 leads to a running time of
$T(n)=O(n)+2\cdot T( n/2)$, which gives $T(n)=O(n\log n)$.
% By the Ham-Sandwich Theorem
%.....
%A ham-sandwich cut can be found in linear time

 Next, we compute an ordering
 \begin{equation}
   \label{ordering}
A_1 \LL A_2 \LL  \cdots \LL A_n.
 \end{equation}
We do this by a standard sorting algorithm in $O(n \log n)$ time,
assuming that the relation~$\LL$ is a linear order.
If, at any time during the sort, we find two segments that
are not comparable by $\LL$, we quit.
Finally, we check condition~\eqref{ordering} in $O(n)$ time.
(This final check is not necessary, if, for example, mergesort is used
as the sorting algorithm.)

As the last step, we check Property~\ref{algo} of Theorem~\ref{x} in
linear time, as outlined above.
\end{proof}

\begin{figure}[hpt]
$$\includegraphics{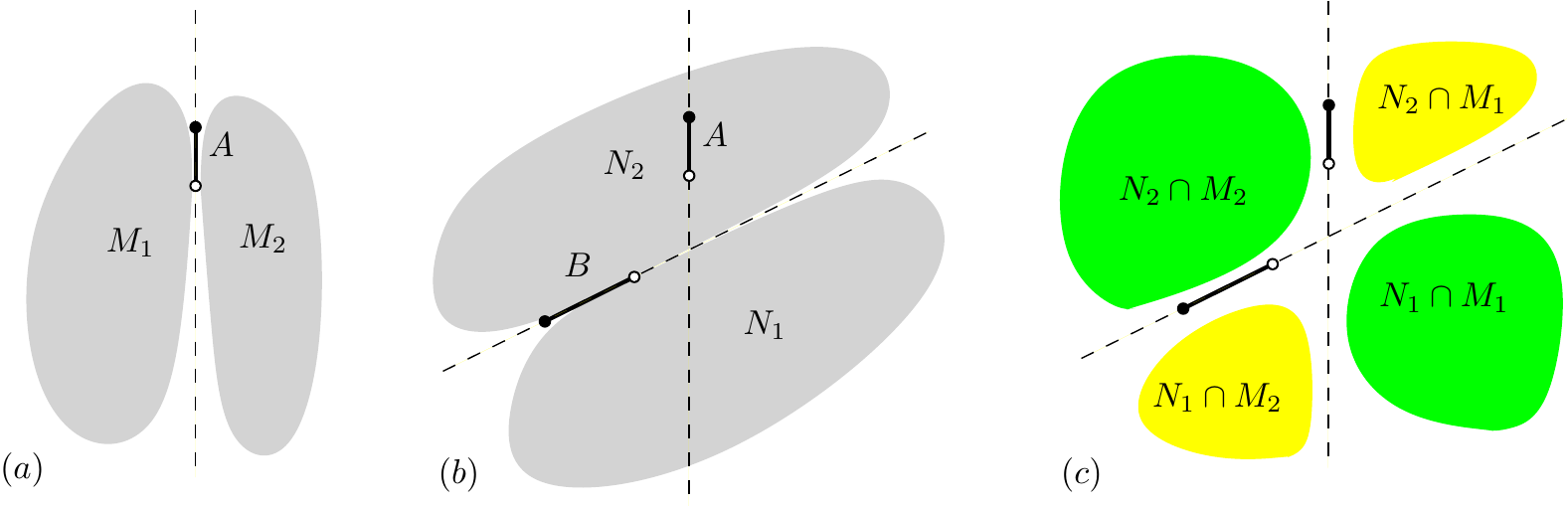}$$
\caption{Separation of M into $6$ overlapping BR-Matchings}
\label{fig:ACM}
\end{figure}

It is also possible to determine in $O(n \log n)$ if a BR-matching $M$
is circular, by an easy divide-and-conquer algorithm. 
Let $A$ and $B$ be two arbitrary segments in $M$. 
Let $M_1=M_A^{L+}$ and $M_2=M_A^{R+}$ (that is, the
% BR-matchings formed by the
 segments that lie to the left or to the right of $A$, including $A$ itself) and
likewise $N_1=M_B^{L+}$ and $N_2=M_B^{R+}$
%$N_1$ and $N_2$ the BR-matchings formed
%by the segments left respectively right of $B$ including $B$ 
(see Fig.~\ref{fig:ACM}; recall that the segments are 
implicitly directed from white to black).
$M_i$ and $N_i$ are linear matchings by Lemma~\ref{lem:foursets}.
Finally, define $Q_1:=(M_2\cap N_2) \cup (M_1\cap N_1)$ and
$Q_2:=(M_1\cap N_2) \cup (M_2\cap N_1)$.

\begin{observation}
A BR-matching $M$ has a chromatic cut if and only if at least one of
the six matchings defined above
has a chromatic cut.
\end{observation}
\begin{proof}
Consider two segments in $M$. Then they must be 
both in one of the matchings $M'$ as defined above.
If they have a chromatic cut, then $M'$ has a chromatic cut.
The other direction is obvious.
\end{proof}

%\begin{observation}%
%	Let $M$ be a circular matching. Then $M_i$ and $N_i$ are linear matchings.
%\end{observation}

\begin{theorem}
	It can be checked in $O(n \log n)$ time whether a BR-matching $M$ is of circular type.
\end{theorem}
\begin{proof}
The algorithm starts to compute the convex hull of $M$.
If all points on $\partial CH(M)$ are of the same color we know it is not a linear matching
and it remains to check if $M$ has no forbidden pattern as in Fig.~\ref{fig:allConfig}~(a)-(b).

% We start with preprocessing the whole matching $M$ by splitting it along the supporting
% line of any segment
We pick any segment $A_0$ and split $M$ along the supporting
line of $A_0$.
We compute the linear order of both parts.
This gives a potential circular order. We remember this order for the remaining part.

The rest of the algorithm works recursively. 
We start by defining $M_1$ and $M_2$ as above to any segment $A$.
Let $B$ be the median of the larger of the $M_i$ with respect to \ord .
The BR-matchings $N_i$ and $Q_i$ are also defined as above.
For the BR-matchings $M_i$ and $N_i$, it can be checked in linear time
if they are linear, because we have already precomputed the order.
As $B$ is the median of the larger of the $M_i$, 
$ n/4\leq | Q_i | \leq 3n/4 , \ \forall i.$ We check recursively if
$Q_1$ and $Q_2$ has 
no chromatic cut. For the running time $T(n)$, we have $T(n)\leq O(n) + \max_{1/4 \leq \alpha \leq 3/4}[T(\alpha n) + T((1-\alpha)n)]$. Thus $T(n) = O(n \log n)$.

If any of these steps in the algorithm
fails, a forbidden configuration is present. 
In this case we just stop and
return that $M$ has a chromatic cut.
Otherwise we return the correct circular order.
\end{proof}

\begin{figure}[h]
		\begin{center}
	\includegraphics{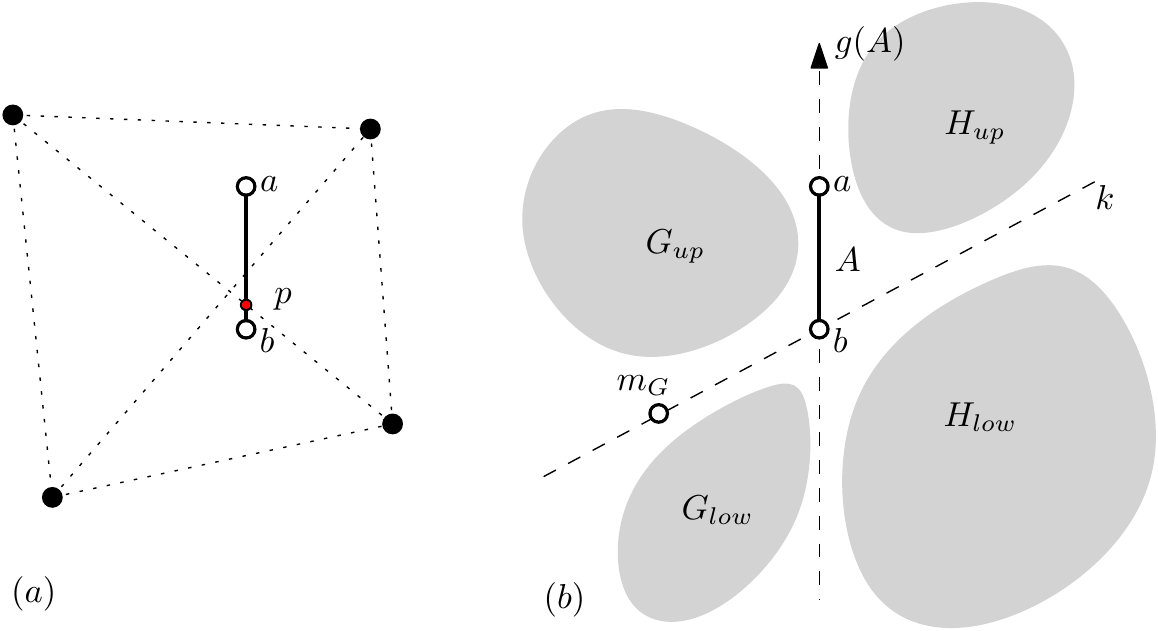}
		\caption{The line supporting line $g(A)$ splits
		the point set $F$ into $G$ and $H$. }
		\label{fig:algoBal}
			\end{center}
	\end{figure}

The last algorithm we want to present is about computing
a balanced line as in Lemma~\ref{lem:cutBal}.
As a preprocessing step we need to find a point on a segment in general
position with respect to the remaining points $F$.
\begin{lemma}\label{lem:lowInt}
  Let $F$ be a point set in the plane and $A=(a,b)$ be a vertical segment 
 such that $F\cup\{a,b\}$ lies in general position
(i.e., no three points	on a line).
	Then the lowest 
	intersection $p$ of $A$ with 
a segment %the \textbf{arrangement of segments} 
formed by
two % the pairs of
 points in $F$
can be computed
 in deterministic $O(n\log n)$ time.
\end{lemma}
\begin{proof}
	Consider the point sets	$G$ and $H$ left and right of $g(A)$. Let $m_G$ be the
	median of the larger set $G$, with respect to the order defined by a ray rotating around $b$. 
	The  line $k$ through $m_G$
	and $b$ defines the four sets
	$G_{\textrm{up}}$, $G_{\textrm{low}}$, $H_{\textrm{up}}$ and $H_{\textrm{low}}$, as in Fig.~\ref{fig:algoBal}b.
	Now any two points defining the lowest intersection with $C$ are either
	in $G_{\textrm{up}}$ and $H_{\textrm{up}}$, or in two opposite sets (i.e., $G_{\textrm{up}}$ and $H_{\textrm{low}}$
	or $G_{\textrm{low}}$ and $H_{\textrm{up}}$). 
	The lowest intersecting segment of $G_{\textrm{up}}$ and $H_{\textrm{up}}$ is the
	convex hull edge of $G_{\textrm{up}}$ and $H_{\textrm{up}}$ intersecting the supporting line  $g(A)$.
	It can be found in linear deterministic time with a subroutine of the  convex hull
	algorithm by Kirkpatrick and Seidel~\cite{doi:10.1137/0215021} or by an algorithm by
	Aichholzer, Miltzow and Pilz~\cite{extremePointJourn}. The second algorithm only uses 
	order type information.
	The two opposite sets are treated recursively. Note that 
	$n/4 \leq \#(G_{\textrm{low}}\cup H_{\textrm{up}}) \leq 3n/4$ and likewise
	$n/4 \leq \#(G_{\textrm{up}}\cup H_{\textrm{low}}) \leq 3n/4$. Therefore, for the running time %$T(n)$ 
	we get $T(n) = O(n) + max_{1/4\le \alpha\le3/4}[T(\alpha n) + T((1-\alpha)n)]$
	and $T(n) = O(n \log n)$.
\end{proof}

\begin{figure}[htp]
		\begin{center}
		\includegraphics{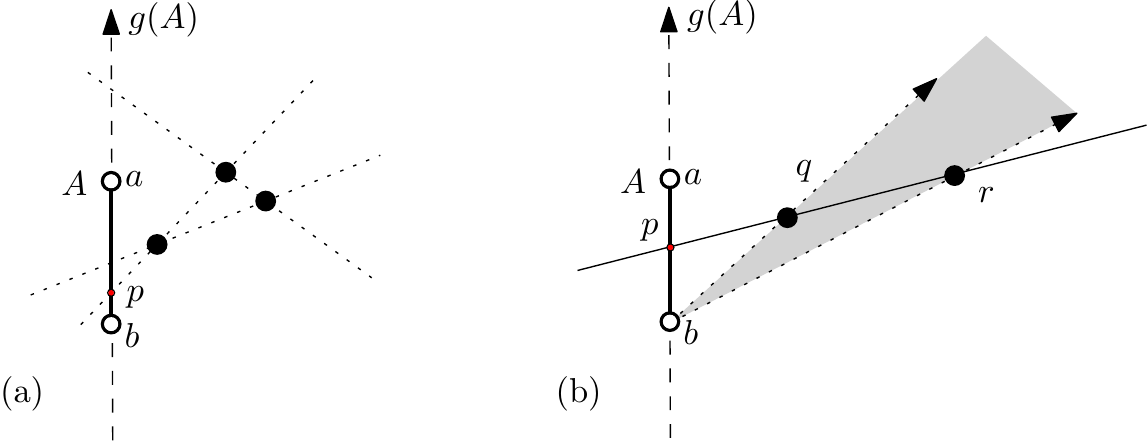}
		\caption{(a) the line arrangement formed by the points in $F$ and the lowest crossing with $A$; (b) the cone with apex $b$ spanned by the minimal pair of points is empty of points of $F$}
		\label{fig:algoMinCross}
			\end{center}
	\end{figure}

For the next Lemma we refer to Fig.~\ref{fig:algoMinCross}.

\begin{lemma}\label{lem:rightArr}
	Let $A$ be a segment and $F$ a point set {right of
          $g(A)$} in general position. Then  the lowest intersection
        $p$ of $A$ with a line through two
%the \textbf{line arrangement} formed by the
%        pairs of
 points in $F$ can be computed in deterministic $O(n\log n)$ time.
\end{lemma}
\begin{proof}
	First consider the points $q,r\in F$ which 
	form the lowest crossing with $A$. 
	We show they are neighbors in the radial order around $b$. 
	Consider the area swept by a ray from $q$ to $r$. 
	If it contained any point $s$ then either the line through $q$ and 
	$s$ or $r$ and $s$ would have a lower intersection with $A$.
%	\smallskip
	
	Thus we merely compute the radial order around $b$
	and for any neighboring pair the intersection point 
	with $A$.
	The running time $T(n)=O(n\log n)$ is dominated by the 
	sorting procedure.
\end{proof}
\begin{theorem}\label{cor:lowInt}
	Let $F$ be a point set in the plane and $A=(a,b)$ be a vertical segment 
 such that $F\cup\{a,b\}$ lies in general position
(i.e., no three points	on a line).
	Then  the lowest 
	intersection of $A$ with a
{line} through two points in $F$
can be computed
 in deterministic $O(n\log n)$ time.
\end{theorem}
\begin{proof}
	Compute the lowest intersection point with a line
	separately for the points  left and right of $A$ according to
	Lemma~\ref{lem:rightArr} and all possible intersections with 
	$A$ by pairs of points on opposite sites of $A$ according to Lemma~\ref{lem:lowInt}.
\end{proof}
\begin{corollary}\label{cor:genPos}
	Given a point set $F$ and a segment $A$ without three points on a line,
	 a point on $A$ in general
	position with respect to $F$
can be computed in $O(n\log n)$ time.
\end{corollary}
\begin{proof}
	Any point between	the lowest intersection and the lower endpoint of $A$ is in 	
	general position with respect to $F$. 
\end{proof}
\begin{lemma}
	Let $M$ be a BR-matching of a point set $F$ in general position 
	and $A$, $B$ be two segments as in Fig.~\ref{fig:allConfig}.
	Then we can compute a balanced line through 
	the interior of $A$ or $B$ in $O(n\log n)$ time.
\end{lemma}
\begin{proof}
	Let $p\in A$ and $q\in B$ be points as in Corollary~\ref{cor:genPos}.
	We know by the proof of Lemma~\ref{lem:cutBal} that a balanced line 
	through $p$ or $q$ exists. 
	%We show how to detect it for say $p$.
	%Consider a \emph{line} $k$ rotating around $p$ and write down the 
	%points in $F$ in the order $k$ passes through them. For any two 
	%points in can be computed in constant time in which order $k$
	%sweeps over them.
	%Count the points right of $k$ in the initial configuration and 
	%keep track of the changes during the rotation. If $k$ is a balanced 
	%line at some stage it will be noticed.
	The algorithm in~\cite{extremePointJourn} % can be adobted to find a halving edge
	can be adapted to find the desired balanced line in through 
	$p$ or $q$ in $O(n)$.
\end{proof}

\begin{remark}
	Once we have $O(n \log n)$ algorithms to test whether a BR-matching
	is linear or circular we automatically receive an algorithm
	to test if a BR-matching has a chromatic cut in $O(n \log n)$.
	Note that both algorithms above can be executed till they 
	find a forbidden configuration. Thus we are able to compute 
	a forbidden configuration also in $O(n\log n)$.
	In the case of linear matchings we compute the linear order and for circular 
	matchings the circular order.
	It is also easy to construct a reference line in linear time, as in 
	the Definition~\ref{def:quasiParallel} of quasi-parallel segments.
	Given a forbidden configuration it is also possible to compute in constant
	time a chromatic cut (i.e., the actual line).
Finally, given a forbidden configuration, we can compute a balanced line intersecting one of the segments.
	In summary, \emph{all} defined terms presented can be computed efficiently.
	%Note that all algorithms, except for computing a BR-matching $M$, 
	%only the order type and not the actual	coordinates are required. 
	%This makes them in a certain sense more robust.
\end{remark}

%%%%%%%%%%%%%%%%%%%%%%%%%%%%%%%%%%%%%%%%%%%%%%%%%%%%%%%%%%%%%%%%%%%%%%%%%%%%%%%%%%%%%%%%%%%%%%%
\section{Open questions, Lower Bounds, etc.}
\label{sec:open}
%%%%%%%%%%%%%%%%%%%%%%%%%%%%%%%%%%%%%%%%%%%%%%%%%%%%%%%%%%%%%%%%%%%%%%%%%%%%%%%%%%%%%%%%%%%%%%%

Our algorithm for testing whether a point set $F$ has a unique
non-crossing BR-matching starts by finding such a BR-matching $M$, in
$O(n\log n)$ time, by repeated ham-sandwich cuts. This algorithm does
not care whether $M$ is unique, and it is in fact the fastest known
algorithm for finding \emph{any} non-crossing BR-matching in an
arbitrary point set.  
%Is there a simpler, or perhaps even a faster
%algorithm for finding $M$ under the assumption that it % $M$
%is unique?
%
%NO. reduction from element distinctness:
%
%
%REPLACEMENT QUESTION:
Is there % a simpler, or perhaps even
a faster
algorithm for checking whether $M$ is unique (necessarily without constructing $M$)?

The paper just read could also be seen as the study of
point sets with certain forbidden patterns.
These particular point sets have a lot of nice geometric structure.
We wonder whether also other forbidden patterns lead
to interesting geometric properties.

Consider $n$ blue, $n$ red and $n$ green points in $\mathbb{R}^3$.
By applying repeatedly ham-sandwich cuts we know that there exists
a noncrossing colorful $3$-uniform geometric matching. (Each edge is represented 
by the convex hull of its vertices.)
Thus we ask for a geometric characterization of point sets with just one such matching.

\section*{Acknowledgments.}
We thank Michael Payne, Lothar Narins and Veit Wiechert for helpful discussions, 
and we thank Stefan Felsner for a hint to the literature.

\bibliography{TillLib_2}

\begin{thebibliography}{10}

\bibitem{AichCompMatch2009}
O.~Aichholzer, S.~Bereg, A.~Dumitrescu, A.~Garc\'{\i}a, C.~Huemer, F.~Hurtado,
  M.~Kano, D.~Rappaport A.~M\'{a}rquez, S.~Smorodinsky, D.~Souvaine,
  J.~Urrutia, and D.~Wood.
\newblock Compatible geometric matchings.
\newblock {\em Computational Geometry: Theory and Applications},
  42(6--7):617--626, 2009.

\bibitem{extremePointJourn}
O.~Aichholzer, T.~Miltzow, and A.~Pilz.
\newblock Extreme point and halving edge search in abstract order types.
\newblock {\em Computational Geometry Theory and Applications}.
\newblock to appear.

\bibitem{abls-bcm-13}
G.~{Aloupis}, L.~{Barba}, S.~{Langerman}, and D.~L. {Souvaine}.
\newblock Bichromatic compatible matchings.
\newblock In Timothy Chan and Rolf Klein, editors, {\em Proc. 29th Ann. Symp.
  Comput. Geometry}, SoCG'13. ACM Press, June 2013.
\newblock To appear. Preprint in
  \href{http://arxiv.org/abs/1207.2375}{arXiv:1207.2375}.

\bibitem{bkos-cgaa-00}
M.~de~Berg, M.~van Kreveld, M.~Overmars, and O.~Schwarzkopf.
\newblock {\em Computational Geometry: Algorithms and Applications}.
\newblock Springer-Verlag, Berlin, Germany, 2nd edition, 2000.

\bibitem{Garcia:2000:LBN:353029.353031}
A.~Garc\'{\i}a, M.~Noy, and J.~Tejel.
\newblock Lower bounds on the number of crossing-free subgraphs of {$K_n$}.
\newblock {\em Comput. Geom. Theory Appl.}, 16(4):211--221, August 2000.

\bibitem{FerranToth2003survey}
F.~Hurtado and Cs.~D. T\'oth.
\newblock Plane geometric graph augmentation: a generic perspective.
\newblock In J.~Pach, editor, {\em Thirty Essays on Geometric Graph Theory},
  volume~29 of {\em Algorithms and Combinatorics}, pages 327--354. Springer,
  2013.

\bibitem{Ishaque:2011:DCG:1998196.1998216}
M.~Ishaque, D.~L. Souvaine, and Cs.~D. Toth.
\newblock Disjoint compatible geometric matchings.
\newblock In {\em Proc. 27th Ann. Symp. on Computational geometry}, SoCG'11,
  pages 125--134. ACM, 2011.

\bibitem{kmpsy-cerpg-08}
L.~Kettner, K.~Mehlhorn, S.~Pion, S.~Schirra, and C.~Yap.
\newblock Classroom examples of robustness problems in geometric computations.
\newblock {\em Comput. Geom. Theory Appl.}, 40:61--78, May 2008.

\bibitem{doi:10.1137/0215021}
D.~Kirkpatrick and R.~Seidel.
\newblock The ultimate planar convex hull algorithm?
\newblock {\em SIAM Journal on Computing}, 15(1):287--299, 1986.

\bibitem{lms-ahsc-94}
C.-Y. Lo, J.~Matou{\v s}ek, and W.~L. Steiger.
\newblock Algorithms for ham-sandwich cuts.
\newblock {\em Discrete Comput. Geom.}, 11:433--452, 1994.

\bibitem{RoteDis1988}
G.~Rote.
\newblock {\em Two solvable cases of the traveling salesman problem}.
\newblock PhD thesis, Technische Universit\"{a}t Graz, 1988.

\bibitem{RoteJan1992}
G.~Rote.
\newblock The {$N$}-line traveling salesman problem.
\newblock {\em Networks}, 22:91--108, 1992.

\bibitem{sw-ncfm-06j}
M.~Sharir and E.~Welzl.
\newblock On the number of crossing-free matchings, cycles, and partitions.
\newblock {\em SIAM J. Comput.}, 36(3):695--720, 2006.

\end{thebibliography}
\bibliographystyle{plain}
	
\end{document}